\documentclass[journal]{IEEEtran}
\usepackage{mathrsfs}
\usepackage{amsmath}
\usepackage{amssymb}
\usepackage{amsthm}
\usepackage{stfloats}
\usepackage{cuted}\stripsep -3pt plus 3pt minus 2pt
\ifCLASSINFOpdf \else \fi

\usepackage[dvips]{graphicx}
\usepackage{algorithm}
\usepackage{color}
\usepackage{booktabs}
\usepackage[normalem]{ulem} 
\usepackage{algorithmic}
\usepackage{slashbox}
\usepackage{diagbox}
\usepackage{subfigure}
\usepackage{xcolor}
\usepackage{setspace}

\newtheorem{Theo}{Theorem}

\begin{document}
\title{Throughput of Hybrid UAV Networks\\ with Scale-Free Topology}
\author{Zhiqing Wei,
		Ziyu Wang, 
		Zeyang Meng,
		Ning Zhang,
		Huici Wu,
		Zhiyong Feng

\thanks{Zhiqing Wei, Zeyang Meng, Huici Wu, and Zhiyong Feng are with Beijing University of Posts and Telecommunications, Beijing, China 100876 (email: \{weizhiqing, mengzeyang, dailywu, fengzy\}@bupt.edu.cn).
	
	Ziyu Wang is with Amazon (China) Holding Company Limited, Beijing, China 100025 (email: ziyuwan@amazon.com).
	
	Ning Zhang is with the Department of Electrical and Computer Engineering, University of Windsor, Windsor, ON, N9B 3P4, Canada. (e-mail: ning.zhang@uwindsor.ca).
	
	Correspondence authors: Ziyu Wang, Huici Wu, and Zhiqing Wei.}}

\maketitle

\begin{abstract}
Unmanned Aerial Vehicles (UAVs) hold great potential to support a wide range of applications due to the high maneuverability and flexibility. Compared with single UAV, UAV swarm carries out tasks efficiently in harsh environment, where the network resilience is of vital importance to UAV swarm. The network topology has a fundamental impact on the resilience of UAV network. It is discovered that scale-free network topology, as a topology that exists widely in nature, has the ability to enhance the network resilience. Besides, increasing network throughput can enhance the efficiency of information interaction, improving the network resilience. Facing these facts, this paper studies the throughput of UAV Network with scale-free topology. Introducing the hybrid network structure combining both ad hoc transmission mode and cellular transmission mode into UAV Network, the throughput of UAV Network is improved compared with that of pure ad hoc UAV network. Furthermore, this work also investigates the optimal setting of the hop threshold for the selection of ad hoc or cellular transmission mode. It is discovered that the optimal hop threshold is related with the number of UAVs and the parameters of scale-free topology. This paper may motivate the application of hybrid network structure into UAV Network.
\end{abstract}
\IEEEpeerreviewmaketitle

\section{Introduction}
In recent years, with the development of the technologies such as mobile communications, artificial intelligence, and automatic control, Unmanned Aerial Vehicles (UAVs) have been widely applied in many areas, such as reconnaissance, disaster rescue, wireless communication and logistics, due to the advantages of high maneuverability and flexibility \cite{Intro1}. The application of UAVs is showing a blowout trend, and the application areas are continuously expanding. However, due to the extremely limited endurance time, power, size, and the harsh environment where UAVs carry out tasks, it is difficult for a single UAV to complete tasks quickly and efficiently. Thus, the collaboration among UAV swarm is required \cite{Intro2}. 
Taking into account the harsh environment in which UAV swarm is applied, the ability of UAV swarm to defend against device faults or other possible attacks and maintain the reliability of services, namely, the network resilience \cite{Intro2.5}, is of vital importance to UAV swarm.
Hence, the network resilience of UAV swarm is a key factor restricting the wide application of UAV swarm \cite{Intro3} \cite{Intro4}.

On the one hand, the network topology has a fundamental impact on the resilience of UAV network. 
Scale-free network topology, as a topology that exists widely in nature, such as Internet and biological swarms, has attracted widespread attention. The power-law distribution of the degree of the nodes is one of the important characteristics of this kind of complex network \cite{ScaleFree_1,ScaleFree_2}. As a typical complex network, the large-scale UAV network is widely studied \cite{ScaleFree_3,ScaleFree_4}, where the large-scale UAV network is a scale-free network. For example, \cite{Intro9} verified that the large-scale network has scale-free characteristic naturally under rules of Reynolds Boids. On the other hand, the study on scale-free structure of UAV network is of great significance to improve the network performance. Because the scale-free structure has great robustness when facing intentional attacks \cite{Intro5, Intro6}. In an asymptotic sense, all network nodes need to be destroyed in order to destroy the scale-free network, which greatly improves the stability and reliability of UAV network. Tran \emph{et al.} \cite{Intro7} and Fan \emph{et al.} \cite{ScaleFree_3} studied the scale-free UAV network and found that randomly removing nodes has little effect on the connectivity of UAV network, so that the network has a higher tolerance for random attacks. To this end, in order to enhance the resilience of UAV swarm, the UAV network topology with scale-free network topology has been studied in depth. In \cite{Intro9}, inspired by the scale-free characteristics of bird flocks enhancing environmental response capabilities, Singh \emph{et al.} studied the scale-free characteristics of UAV swarm to improve the survivability of UAV swarm.

On the other hand, the network throughput has a correlation with network resilience.
The increase of network throughput will decrease the congestion probability and the communication delay.
As a result, the UAV swarm has a short response time to the interruption of the network, and the network resilience is improved.
In order to alleviate congestion and reduce network response time, 
Defense Advanced Research Projects Agency (DARPA) \cite{Intro10} released a project called Content-Based Mobile Edge Networking (CBMEN) to effectively improve network throughput and reduce delay.
Liu \emph{et al.} \cite{Intro10.5} modeled the relation between network resilience and network throughput, and maximized the throughput to improve the fault recovery ability of the network.

Therefore, it is shown that both network topology and network throughput can comprehensively affect the resilience of UAV network. In terms of the throughput of UAV network, Yuan \emph{et al.} studied the impact of the mobility of
UAVs on the link throughput of UAV Network \cite{Intro11}. Li \emph{et al.} \cite{Intro13} studied the throughput of air-to-air links and the multiple access channel (MAC) throughput of UAV network. Chetlur \emph{et al.} \cite{Intro14} studied the outage probability of the link of three-dimensional (3D) UAV network, and analyzed the resilience of the UAV network from the perspective of reliability. Gao \emph{et al.} \cite{Intro15} enhanced the throughput of UAV Network by optimizing the deployment of UAVs in 3D space. We studied the throughput of 3D UAV network in \cite{Intro16}, discovering that the UAV network throughput is a function of the path loss factor in 3D space, the number of nodes, the factor of contact concentration, and so on. However, the throughput of UAV network with scale-free topology was seldom studied. The research on the wireless network with scale-free topology is firstly reported in \cite{Intro17},
where the degree of node follows a power-law distribution, which is the feature of scale-free network. 
\cite{Intro17} has found the relation between network throughput and node degree. By assigning independent spectrum resources to nodes with high degrees, the throughput of the UAV network is improved. 
{\cite{Intro18} and \cite{Intro18.5} are preliminary works} of \cite{Intro17}, where the distribution of node's degrees follows uniform distribution. We studied the throughput of 3D scale-free network in \cite{Intro19}. Compared with \cite{Intro17}, \cite{Intro19} studied the impact of 3D network topology on the throughput of scale-free network. Besides, the optimal threshold of the degree is studied, and the separated resources are allocated to nodes with the degree larger than the threshold to enhance the network throughput.

To warp up, the scale-free topology improves the resilience of UAV network, and the enhancement of the throughput of UAV network is essential to improve the resilience of UAV network. In order to improve the throughput of UAV network, a hybrid UAV network is formed with the cooperation between UAV network and ground cellular network in this paper. Hybrid network is a combination of ad hoc network and cellular network \cite{Intro20}. As shown in Fig. 1, when the source node and the sink node are far away, the data can be transmitted with cellular mode. On the contrary, when the source node and the sink node are close, the data can be transmitted via ad hoc mode. Kumar \emph{et al.} proposed this model earlier in \cite{Intro21} and proved the improvement of network throughput through the hybrid network model. 

In this paper, the throughput of the hybrid UAV network with scale-free topology is studied. 
The contributions of this paper are as follows. 
\begin{itemize}
	{\item The 3D model is applied in this paper, which is more realistic than the two-dimensional (2D) model. The 3D model has mainly two differences compared with 2D model. Firstly, the relative relationship between UAV and base station (BS) is more practical. Considering the flight capability and spatial distribution characteristics of the UAV, the BS is located on the 2D plane, namely the ground, and the UAV is located in a the cubic 3D space, rather than simply assuming that the UAV and the BS are deployed on the same plane. Secondly, the difference on dimension causes the difference on power and segmentation of theoretical results between two models.}
	
	{\item The three dimensions of the scale-free characteristics are considered in the hybrid network, i.e., the probability of source nodes selecting contact group members follows the power-law distribution corresponding to the distance, the probability of source nodes communicating with contact group members follows the power-law distribution corresponding to the distance, and the number of contact group members follows a power-law distribution. Compared with the existing researches on the throughput of hybrid scale-free networks, such as \cite{Intro18.5}, the scale-free characteristic of the number of members in the contact group, i.e., the power-law exponent $\gamma$, is taken into consideration in this paper.}
\end{itemize}

This paper is organized as follows.
In Section \uppercase\expandafter{\romannumeral2},
the network model of hybrid UAV network with scale-free topology is introduced.
In Section \uppercase\expandafter{\romannumeral3}, the throughput of hybrid UAV Network with scale-free topology is derived.
The numerical results of the analytical results are shown in Section \uppercase\expandafter{\romannumeral4}. 
Finally, in Section \uppercase\expandafter{\romannumeral5}, we summarize this paper.

\section{Network Model}

The hybrid network is the combination of ad hoc network and cellular network.
The cellular network serves as backbone network.
As illustrated in Fig. \ref{fig_hybrid}, in hybrid UAV network, the BSs are distributed on ground, and the UAVs are uniformly distributed in the unit cube{\footnote{The square-area assumption and the division method of a 2D plane is described in \cite{FootNote_1}, which is equal to a Voronoi tessellation satisfying Remark 5.6 in \cite{FootNote_1}. The division method guarantees that there is at least one node in each small square when the number of nodes $n$ tends to infinity. Similarly, the division method of 3D unit cube in this paper can be analogized from the result mentioned above, which is also applied in \cite{FootNote_2}.}}.
When the distance between source and destination is small,
the information flow goes through ad hoc mode. However, when the distance between source and destination is
large, the information flow goes through
cellular mode.

\begin{figure}[!ht]
\centering
\includegraphics[width=0.49\textwidth]{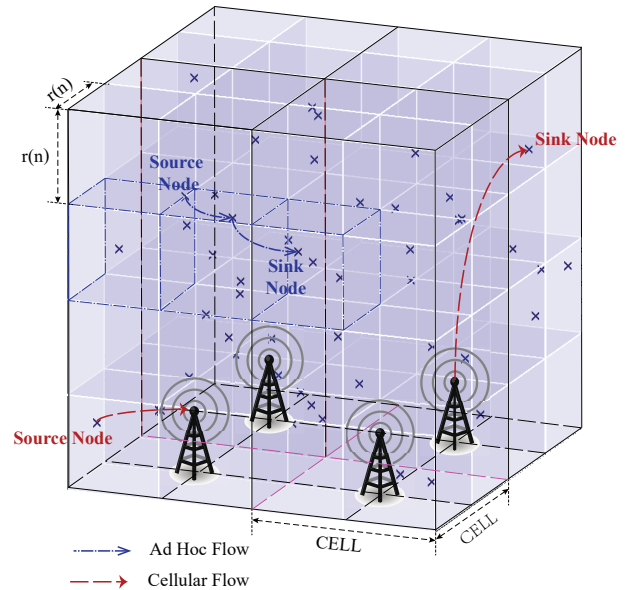}
\caption{Multi-hop routing strategy from source to destination.}
\label{fig_hybrid}
\end{figure}

\subsection{Communication model}

\subsubsection{Interference model}
\label{sec_interference_model}

In the UAV network, $n$ UAVs are uniformly distributed in the unit cube.
The unit cube is divided into small cubes
with side length
$\Theta \left( {{{\left( {\frac{{\log n}}{n}} \right)}^{{\textstyle{1 \over 3}}}}} \right)$\footnote{In this paper, $f( n ) = O(g(n))$ means that $\mathop {\lim }\limits_{n \to \infty } \frac{{f(n)}}{{g(n)}} < \infty $; $f\left( n \right) = \Omega (g(n))$ means that $g( n ) = O(f(n))$; $f( n ) = \Theta (g(n))$ means that $f( n ) = O(g(n))$ and $g( n ) = O(f(n))$, which is also denoted by $f(n) \equiv g(n)$.}.
{ According to Fig. \ref{fig_hybrid}, each small cube contains at least one node with high probability (\emph{w.h.p.})	
if the transmission range $r(n)$ between the two nodes is as follows \cite{19}.}

\begin{equation}\small
\label{eq_1}
r\left( n \right) = \Theta \left( {{{\left( {\frac{{\log n}}{n}} \right)}^{{\textstyle{1 \over 3}}}}} \right).
\end{equation}

We apply the \textsl{protocol model} \cite{02} for interference management.
Assuming that the 3D
Cartesian coordinates of the nodes $i$, $j$, $k$
are ${X_i}$, ${X_j}$, ${X_k}$ respectively,
two nodes can communicate successfully when
\begin{equation}\small
\left| {{X_i} - {X_j}} \right| < r\left( n \right),
\end{equation}
and the other nodes that transmit on the same frequency band satisfy the condition
\begin{equation}\small
\left| {{X_k} - {X_j}} \right| > \left( {1 + \Delta } \right)r\left( n \right),
\end{equation}
where $\Delta  > 0$ is the guard zone factor.

\subsubsection{Multiple access control}

In order to avoid multiple access interference, time division multiple access (TDMA) is adopted.
Suppose that the side length of each small cube is ${c_1}r\left( n \right)$, where
${c_1}$ is a constant smaller than $1$
to ensure that all nodes in the neighboring cubes are within the transmission range.
According to the interference model in Section \ref{sec_interference_model},
only the nodes within the intervals of $M$ cubes
are allowed to communicate simultaneously, where $M \ge \frac{{2 + \Delta }}{{{c_1}}}$.
Then ${M^3}$ cubes become a cluster,
and the cubes in the entire cluster are traversed in ${M^3}$ time slots
in a round-robin scheduling method.
The TDMA scheme in this paper is denoted by $M^3$-TDMA scheme.
As shown in Fig. \ref{TDMA-model},
the green cubes are located in different clusters,
and the nodes in these cubes can transmit data at the same time.

{Note that the analytical results under such protocol are also applicable to other multiple access control (MAC) protocol. Some studies have proved that MAC protocol will not affect the throughput scaling law. For example, \cite{MAC_1} studied the throughput bound scaling law of ad hoc network under Carrier Sense Multiple Access with Collision Avoid (CSMA/CA) protocol, which proved that the ad hoc network with CSMA/CA protocol has the same throughput  scaling law as the ad hoc network with TDMA protocol. }

\begin{figure}[!ht]
\centering
\includegraphics[width=0.49\textwidth]{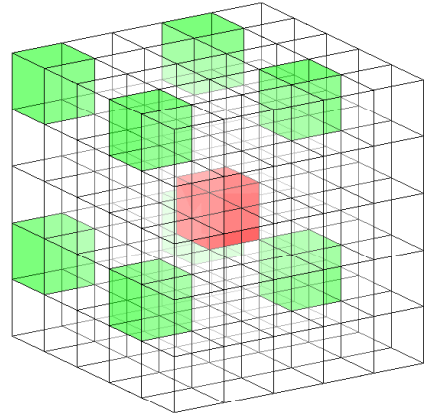}
\caption{The illustration of TDMA model where nodes in green cubes can transmit data at the same time.
In this figure, we assume that the source node is located in the red cube.}
\label{TDMA-model}
\end{figure}

\subsubsection{Data flows}
As illustrated in Fig. \ref{fig_hybrid}, BSs are distributed in a unit square on ground. The unit square is divided into $m =\Theta \left( {{{\left( {\frac{{\log n}}{n}} \right)}^{-{\textstyle{2 \over 3}}}}} \right)$ cells. And there is a BS at the center of each cell.
The volume of each cell is $1/m$. A UAV is associated with the nearest BS.
The BSs are connected via optical fiber with high throughput. Hence, there are no throughput limitations among the BSs.
In Fig. \ref{fig_hybrid}, there are two kinds of information flows, namely, ad hoc flow adopting multi-hop transmission and cellular
flow adopting BSs to transmit data.
The total bandwidth $W$ bits per second (bps) is divided into two parts, with $W_a$ bps allocated to ad hoc flows and $W_c$ bps allocated to cellular flows. Thus, we have

\begin{equation}\small
\label{eq_4}
W = {W_a} + {W_c}.
\end{equation}

\subsubsection{Routing scheme}

The $L$-routing scheme \cite{9.5} is applied in this paper.
If the number of hops from source to destination is smaller than $L$,
the ad hoc transmission mode is adopted.
Otherwise, the information is transmitted via cellular mode.
For the ad hoc information flow denoted by blue line in Fig. \ref{fig_hybrid},
the straight line routing is adopted,
the information is transmitted from source to
destination through the cubes passed through by the line
connecting source and destination.
For cellular flow, the source transmits data to the nearby BS.
Then, the data is transmitted to the BS associated to destination
and finally forwarded to the destination.

\subsection{Scale-free network model}
The network model of scale-free network consists of the distance based contact group model, the communication model of nodes and the distribution of the number of members in contact group.
The contact group of node $S$ is a collection of destination nodes that communicate with node $S$ over a period of time.
The distance based contact group construction describes the probability model of source node with specific contact group.
The communication model describes the communication probability of nodes in contact groups.
The number of members in contact group describes the probability model of the number of contact group members.

\subsubsection{Distance based contact group construction}
Source node $S$ selects any other nodes as the member
of its contact group $G$ with a power-law distribution
probability \cite{10}.
With ${d_i}$ denoting the distance between $S$ and node ${o_i}$,
the probability that ${o_i}$ is selected as a member of contact group follows power-law distribution as follows.
\begin{equation}\small
P\left( {D = {o_i}} \right) = {d_i}^{- \alpha },
\end{equation}
where $\alpha$ is a factor representing the concentration of the network, which is named as \emph{concentration factor} in this paper.
{ When $\alpha$ is large, the selection probability of contact group members attenuates greatly with distance, and members in the contact group tend to be located near the source node.}
The selection of the member of contact group is an independent process.
Thus, the probability that $G$ consists of nodes
${o_{{g_1}}},{o_{{g_2}}},...,{o_{{g_q}}}$ is \cite{Intro17}
\begin{equation}\small
\Pr \left( {G = \left\{{{o_{{g_1}}},...,{o_{{g_q}}}} \right\}} \right) = \frac{{d_{{g_1}}^{- \alpha }...d_{{g_1}}^{- \alpha }}}{{\sum\nolimits_{1 \le {i_1} < ... < {i_q} \le n} {d_{{i_1}}^{- \alpha }...d_{{i_q}}^{- \alpha }} }}.
\end{equation}

The denominator is an elementary symmetric polynomial
and can be denoted as \cite{Intro17}

\begin{equation}\small
{\sigma _q}\left( {{{\bf{d}}_{\bf{n}}}} \right) = \sum\nolimits_{1 \le {i_1} < ... < {i_q} \le n} {d_{{i_1}}^{- \alpha }...d_{{i_q}}^{- \alpha }},
\end{equation}
where ${{\bf{d}}_{\bf{n}}} = \left( {d_1^{- \alpha },...,d_n^{- \alpha }} \right)$ is
an $n$-dimensional vector.
Calculating the synthesis of all combinations,
the probability of an arbitrary particular node ${o_k}$ being
a member of $G$ is denoted by \cite{Intro17}
\begin{equation}\label{E1}\small
\Pr \left( {{o_k} \in G} \right) = \frac{{d_k^{- \alpha }{\sigma _{q - 1}}\left( {{\bf{d}}_{\bf{n}}^{\overline {\bf{k}} }} \right)}}{{{\sigma _q}\left( {{{\bf{d}}_{\bf{n}}}} \right)}},
\end{equation}
where ${\bf{d}}_{\bf{n}}^{\overline {\bf{k}} }$ is the
($n-1$)-dimensional vector except of the $k$-th element $d_k^{- \alpha }$.

\subsubsection{Communication model of nodes}

With the contact group established, the probability of source node $S$ choosing the destination node inside the contact group $G$ to
communicate also follows power-law distribution.
The probability of node $o_i$ selected to be
the destination is ${d_i}^{- \beta }$,
where the factor $\beta$ reveals the communication activity level {of the contact group}, which is named as \emph{communication activity factor} in this paper.
Thus, the probability that ${o_k}$ is the destination node $D$ in $G$ is
\begin{equation}\label{E2}\small
\Pr \left( {D = {o_k}\left| {{o_k} \in G} \right.} \right) = \frac{{d_k^{- \beta }}}{{\sum\nolimits_{i = 1}^q {d_{g_i}^{- \beta }} }} = \frac{{d_k^{- \beta }}}{{{\sigma _1}\left( {{{\bf{d}}_{\bf{q}}}} \right)}},
\end{equation}
where ${{\bf{d}}_{\bf{q}}} = \left( {d_{{g_1}}^{- \beta },...,d_{{g_q}}^{- \beta }} \right)$.
{When $\beta$ is large, the probability of communication destination selection decreases greatly with distance, and the source node tends to communicate with the node at a close location.}

\subsubsection{Number of members in contact group}
The number of members in the contact group, which is the degree of a node, is denoted by $d$.
Then, the probability density function (PDF) of $d$
follows a power-law distribution as follows.

\begin{equation}\small
P\left( {d = q} \right) \propto {q^{- \gamma }},
\end{equation}
where $q$ is a positive integer, $\gamma$ is the power-law exponent, which is named as \emph{clustering factor} in this paper. $A \propto B$ means that A is proportional to B. {When $\gamma$ is large, the number of contact group members is small.}

Assume that each source node $S$
has a contact group $G$ and the number of $G$'s members is a random variable $Q$.
The probability that $G$ has
$q\left( {q = 1,2,...,n} \right)$ members is \cite{Intro17}
\begin{equation}\label{E_total1}\small
\Pr \left( {Q = q} \right) = \frac{{{q^{- \gamma }}}}{{\sum\nolimits_{q = 1}^{n - 1} {{q^{- \gamma }}} }} = \frac{{{q^{- \gamma }}}}{{{\sigma _1}\left( {\bf{q}} \right)}},
\end{equation}
where ${\sigma _1}\left( {\bf{q}} \right)$ is an elementary symmetric polynomial, with ${\bf{q}} = \{1^{- \gamma}, 2^{- \gamma}, ..., (n-1)^{-\gamma} \}$.

\section{Throughput of Hybrid UAV Network}

The per-node throughput of ad hoc mode is denoted by $\lambda _a^n$ bps.
The per-node throughput of cellular mode is denoted by $\lambda _c^n$ bps.
Assuming that the number of ad hoc flows and cellular flows is $N_a$ and $N_c$ respectively,
the network throughput is 
\footnote{{The per-node throughput is denoted by superscript `$n$', while the throughput of the network has no superscript. The subscript `$a$' denotes ad hoc mode and the subscript `$c$' denotes cellular mode.}}
\begin{equation}\small
\lambda = {\lambda _a} + {\lambda _c} = {N_a}\lambda _a^n + {N_c}\lambda _c^n.
\end{equation}

\subsection{Network throughput of cellular mode}

The network throughput of cellular mode is related with the
number of BSs $m$ and the bandwidth $W_c$. We have the following theorem.

\begin{Theo}
The network throughput of cellular mode ${\lambda _{c}}$ satisfies the following equality.
\begin{equation}\small
{\lambda _c}  = \Theta \left( {m{W_c}} \right).
\end{equation}

\end{Theo}
\begin{proof}
{ According to the bandwidth allocation strategy in \eqref{eq_4}, the network throughput of each cell has upper bound $\lambda _c^m = O\left( {{W_c}} \right)$.}
Assuming that there are ${x_{cells}}$ cells sharing the same bandwidth $W_c$,
the lower bound of the throughput of each cell is $\lambda _c^m = \Omega \left( {{{{W_c}} \mathord{\left/
 {\vphantom {{{W_c}} {{x_{cells}}}}} \right.
 \kern-\nulldelimiterspace} {{x_{cells}}}}} \right)$, where ${x_{cells}}$ is a constant that is
independent with $n$ and $m$ \cite{22}.
Hence, the network throughput of each cell is $\lambda _{c}^m = \Theta \left( {{W_c}} \right)$.
Because there are totally $m$ cells, the network throughput contributed by
cellular mode is ${\lambda _{c}} = m\lambda _c^m = \Theta \left( {m{W_c}} \right)$.
\end{proof}

\subsection{Network throughput of ad hoc mode}\label{NC_of_adhoc}
The network throughput of ad hoc mode depends on the average number of hops of ad hoc flows passing through each small cube $E[F]$, where $F$ is the number of ad hoc flows contained in each small cube.
\begin{Theo}
	The per-node throughput of ad hoc mode satisfies the following equation.
	\begin{equation}\small
		\lambda_{a}^{n} \equiv \Theta\left(\frac{W_{a}}{E[F] M^{3}}\right)=\Theta\left(\frac{W_{a}}{E[F]}\right).
	\end{equation}
\end{Theo}
\begin{proof}
	As mentioned above, the side length of the small cube is $c_{2} r(n)=\Theta\left((\log n / n)^{\frac{1}{3}}\right)$. Supposing that the number of hops from source to destination is $X$, where $X$ is a random variable, the average number of hops of one ad hoc flow is $E[X]$. Therefore, the average number of hops of all the ad hoc flows is ${N_a}E\left[ X \right]$. Because of the random distribution of nodes, each cube contains $E[F]=N_{a} E[X] V$ transmission flows, where $V = {\left( {{c_2}r\left( n \right)} \right)^3}$ is the volume of the small cube. According to the \emph{Multiple Access Protocol} in Section \uppercase\expandafter{\romannumeral2},
	the average bandwidth of each slot is $\frac{W_a}{M^3}$. Therefore, the per-node throughput $\lambda _a^n$ of each node is
	\begin{equation}\label{eq_16}\small
		\lambda _a^n \equiv \Theta \left( {{{{W_a}} \over {E\left[ F \right]M^3}}} \right){\rm{}} = \Theta \left( {{{{W_a}} \over {E\left[ F \right]}}} \right).
	\end{equation}
\end{proof}
Therefore, the network throughput contributed by ad hoc mode is ${\lambda _a} = {N_a}\lambda _a^n$.

\subsection{Node and flow classification}
The structure of cubes with $x$ hops away from a source node is a octahedron, as the shaded cubes illustrated in Fig. \ref{Cubes}, which consists of $4{x^2} + 2$ cubes.
\begin{figure}[!ht]
	\centering
	\includegraphics[width=0.4\textwidth]{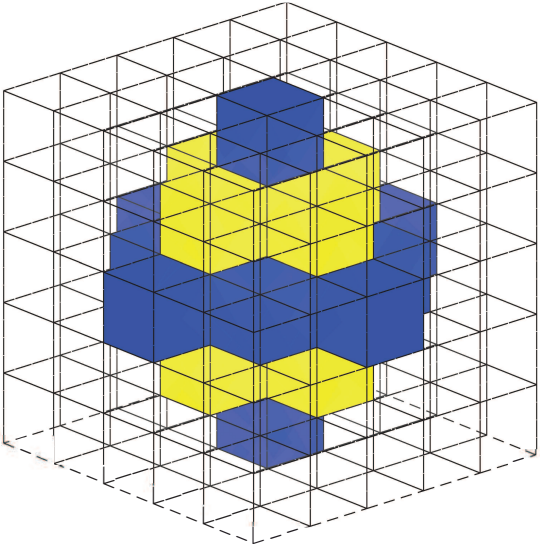}
	\caption{Cubes with $x$ hops away from the source node, where $x = 2$}
	\label{Cubes}
\end{figure}
$\Pr \left( {X = x} \right)$ represents the probability that the distance from the destination node $D$ to the source node $S$ is $x$ hops. According to (6) in \cite{Intro19}, we have
\begin{equation}\small
	\Pr \left( {X = x} \right) = \sum\nolimits_{l = 1}^{4{x^2} + 2} {\sum\nolimits_{{o_k} \in {c_l}} {\Pr \left( {D = {o_k}} \right)} },
\end{equation}
where $c_l$ is the set of the nodes in the cube that is $x$ hops away from the source node and $o_k$ is the destination node within it. Because the nodes are randomly distributed, the probability that any node is located in the cube is ${r^3}\left( n \right)$. { Therefore, the number of nodes contained in $c_l$ is $n{r^3}\left( n \right)$ on average.} Thus we have
\begin{equation}\label{eq_17}\small
	\Pr \left( {X = x} \right) = \sum\limits_{l = 1}^{4{x^2} + 2} {n{r^3}\left( n \right)\Pr \left( {D = {o_k}} \right)}.
\end{equation}

Because of the same power-law distribution as \cite{Intro19}, we use the same symbol as \cite{Intro19}, where $\alpha$ represents the concentration of the network, $\beta$ reveals the communication activity level, and $\gamma$ is the clustering factor.
According to (\ref{E2}) (\ref{eq_17}) in this paper and (7) in \cite{Intro19}, we have the following probability
\begin{equation}\label{eq_18}\small
	\Pr \left( {X = x} \right) = \sum\limits_{l = 1}^{4{x^2} + 2} {\sum\limits_{{o_k} \in {c_l}} {\sum\limits_{q = 1}^{n - 1} {{{{q^{- \gamma }}d_k^{- \alpha  - \beta }{\sigma _{q - 1}}\left( {{\bf{d}}_{\bf{n}}^{\overline {\bf{k}} }} \right)} \over {{\sigma _1}\left( {\bf{q}} \right){\sigma _1}\left( {{{\bf{d}}_{\bf{q}}}} \right){\sigma _q}\left( {{{\bf{d}}_{\bf{n}}}} \right)}}} } }.
\end{equation}

When the destination node is in the regions that are within $L$ hops to the source node $S$, the data flow is forwarded with ad hoc mode. The probability that a flow is an ad hoc flow is denoted by ${\rm Pr}^a$, then we have
\begin{equation}\label{eq_19}\small
	\begin{aligned}
		{{\rm Pr}^a} &= \sum\limits_{x = 1}^L {\Pr \left( {X = x} \right)}  \\&= \sum\limits_{x = 1}^L {\sum\limits_{l = 1}^{4{x^2} + 2} {\sum\limits_{{o_k} \in {c_l}} {\sum\limits_{q = 1}^{n - 1} {{{{q^{- \gamma }}d_k^{- \alpha  - \beta }{\sigma _{q - 1}}\left( {{\bf{d}}_{\bf{n}}^{\overline {\bf{k}} }} \right)} \over {{\sigma _1}\left( {\bf{q}} \right){\sigma _1}\left( {{{\bf{d}}_{\bf{q}}}} \right){\sigma _q}\left( {{{\bf{d}}_{\bf{n}}}} \right)}}} } } }.
	\end{aligned}
\end{equation}

When the destination node is in the regions that are more than $L$ hops to the source node, the data flow is forwarded with cellular mode. According to \cite{Intro19}, the maximum number of hops of each flow is $\Theta \left( {{r^{- 1}}\left( n \right)} \right)$. The probability that a flow is a cellular flow is denoted by ${\rm Pr}^c$. Then we have
\begin{equation}\small
	{{\rm Pr} ^c} = \sum\limits_{x = L + 1}^{{r^{- 1}}\left( n \right)} {\sum\limits_{l = 1}^{4{x^2} + 2} {\sum\limits_{{o_k} \in {c_l}} {\sum\limits_{q = 1}^{n - 1} {{{{q^{- \gamma }}d_k^{- \alpha  - \beta }{\sigma _{q - 1}}\left( {{\bf{d}}_{\bf{n}}^{\overline {\bf{k}} }} \right)} \over {{\sigma _1}\left( {\bf{q}} \right){\sigma _1}\left( {{{\bf{d}}_{\bf{q}}}} \right){\sigma _q}\left( {{{\bf{d}}_{\bf{n}}}} \right)}}} } } }.
\end{equation}

The feature of scale-free topology shows that a few
nodes have a large number of associated nodes.
Thus, a threshold $q_0$ of node degree is chosen that
classifies all the nodes into two classes.
The nodes whose degree $q > q_0$ are leader nodes,
and the nodes whose degree $q \le q_0$ are normal nodes.
We define ${\rm Pr}_1^a$ as the probability for
leader nodes transmitting with ad hoc mode and ${\rm Pr}_1^c$ for
leader nodes transmitting with cellular mode, then
\begin{equation}\small
	\label{eq_21}
	{\rm Pr}_1^a = \sum\limits_{x = 1}^L {\sum\limits_{l = 1}^{4{x^2} + 2} {\sum\limits_{{o_k} \in {c_l}} {\sum\limits_{q = {q_0} + 1}^{n - 1} {{{{q^{- \gamma }}d_k^{- \alpha  - \beta }{\sigma _{q - 1}}\left( {{\bf{d}}_{\bf{n}}^{\overline {\bf{k}} }} \right)} \over {{\sigma _1}\left( {\bf{q}} \right){\sigma _1}\left( {{{\bf{d}}_{\bf{q}}}} \right){\sigma _q}\left( {{{\bf{d}}_{\bf{n}}}} \right)}}} } } },
\end{equation}
\begin{equation}\small
	\label{eq_22}
	{\rm Pr}_1^c = \sum\limits_{x = L + 1}^{{r^{- 1}}\left( n \right)} {\sum\limits_{l = 1}^{4{x^2} + 2} {\sum\limits_{{o_k} \in {c_l}} {\sum\limits_{q = {q_0} + 1}^{n - 1} {{{{q^{- \gamma }}d_k^{- \alpha  - \beta }{\sigma _{q - 1}}\left( {{\bf{d}}_{\bf{n}}^{\overline {\bf{k}} }} \right)} \over {{\sigma _1}\left( {\bf{q}} \right){\sigma _1}\left( {{{\bf{d}}_{\bf{q}}}} \right){\sigma _q}\left( {{{\bf{d}}_{\bf{n}}}} \right)}}} } } }.
\end{equation}

Similarly, we define ${\rm Pr}_2^a$ as the probability for normal nodes transmitting with ad hoc mode and ${\rm Pr}_2^c$ as the probability for normal nodes transmitting with cellular mode, then
\begin{equation}\label{eq_pr2a}\small
	{\rm Pr}_2^a = \sum\limits_{x = 1}^L {\sum\limits_{l = 1}^{4{x^2} + 2} {\sum\limits_{{o_k} \in {c_l}} {\sum\limits_{q = 1}^{{q_0}} {{{{q^{- \gamma }}d_k^{- \alpha  - \beta }{\sigma _{q - 1}}\left( {{\bf{d}}_{\bf{n}}^{\overline {\bf{k}} }} \right)} \over {{\sigma _1}\left( {\bf{q}} \right){\sigma _1}\left( {{{\bf{d}}_{\bf{q}}}} \right){\sigma _q}\left( {{{\bf{d}}_{\bf{n}}}} \right)}}} } } },
\end{equation}
\begin{equation}\label{eq_pr2c}\small
	{\rm Pr}_2^c = \sum\limits_{x = L + 1}^{{r^{- 1}}\left( n \right)} {\sum\limits_{l = 1}^{4{x^2} + 2} {\sum\limits_{{o_k} \in {c_l}} {\sum\limits_{q = 1}^{{q_0}} {{{{q^{- \gamma }}d_k^{- \alpha  - \beta }{\sigma _{q - 1}}\left( {{\bf{d}}_{\bf{n}}^{\overline {\bf{k}} }} \right)} \over {{\sigma _1}\left( {\bf{q}} \right){\sigma _1}\left( {{{\bf{d}}_{\bf{q}}}} \right){\sigma _q}\left( {{{\bf{d}}_{\bf{n}}}} \right)}}} } } }.
\end{equation}

\subsection{Analysis of average number of hops}

As for ${\rm Pr}_1^a$ and ${\rm Pr}_1^c$, we have the following theorem.

\begin{Theo}
	$\operatorname{Pr}_{1}^{a}$ and $\operatorname{Pr}_{1}^{c}$ satisfy the following equations.
	\begin{equation}\label{eq_pr1a}\small
		\operatorname{Pr}_{1}^{a} \equiv\left\{\begin{array}{ll}
			\Theta\left(r^{3}(n) L^{3-\beta}\right) & 0 \leq \beta<3 \\
			\Theta\left(r^{3}(n) \ln L\right) & \beta=3 \\
			\Theta\left(r^{3}(n)\right) & \beta>3
		\end{array}\right.
	\end{equation}
\begin{equation}\label{eq_pr1c}\small
	\operatorname{Pr}_{1}^{c} \equiv\left\{\begin{array}{ll}
		\Theta\left(r^{\beta}(n)-r^{3}(n) L^{3-\beta}\right) & 0 \leq \beta<3 \\
		\Theta\left(r^{3}(n) \ln \left(\frac{1}{\operatorname{Lr}(n)}\right)\right) & \beta=3 \\
		\Theta\left(r^{3}(n)\right) & \beta>3
	\end{array}\right.
\end{equation}
\end{Theo}

\begin{proof}
Please refer to Appendix A.
\end{proof}

\begin{Theo}
The orders of $\Pr _2^a$ and $\Pr _2^c$ are listed in {Table \ref{tab_Pr2a} and Table \ref{tab_Pr2c}}, respectively.
\end{Theo}

\begin{proof}
Please refer to Appendix B.
\end{proof}

\begin{table*}[!ht]
	\caption{The probability of normal ad hoc flow $\operatorname{Pr}_{2}^{a}$}
	\renewcommand{\arraystretch}{1.3} 
	\label{tab_Pr2a}
	\centering{
	\scalebox{0.75}{
		\begin{tabular}{|c|c|c|c|}

		\hline
		\diagbox{$\alpha + \beta$}{$\alpha$} & $0 \leq \alpha<3$ & $\alpha=3$ & $\alpha>3$\\
		\hline
		\multicolumn{4}{|c|}{$\gamma>1$}\\
		\hline
		$0 \leq \alpha+\beta<3$ & $\Theta\left(r^{3-\alpha}(n) L^{3-\alpha-\beta}\right)$ & -- & --\\
		\hline
		$\alpha+\beta=3$ & $\Theta\left(r^{3-\alpha}(n) \ln L\right)$ & $\Theta\left(\log _{r^{-1}(n)} L\right)$ & --\\
		\hline
		$\alpha+\beta>3$&$\Theta\left(r^{3-\alpha}(n)\right)$&$\Theta\left(\ln ^{-1}\left(r^{-1}(n)\right)\right)$&$\Theta(1)$\\
		\hline
		\multicolumn{4}{|c|}{$0 \leq \gamma \leq 1$}\\
		\hline
		$0 \leq \alpha+\beta<3$ & $\Theta\left(n^{\gamma-1} r^{3-\alpha}(n) L^{3-\alpha-\beta}\right)$ & -- & --\\
		\hline
		$\alpha+\beta=3$&$\Theta\left(n^{\gamma-1} r^{3-\alpha}(n) \ln L\right)$&$\Theta\left(n^{\gamma-1} \log _{r^{-1}(n)} L\right)$& --\\
		\hline
		$\alpha+\beta>3$ & $\Theta\left(n^{\gamma-1} r^{3-\alpha}(n)\right)$ & $\Theta\left(n^{\gamma-1} \ln ^{-1}\left(r^{-1}(n)\right)\right)$ & $\Theta\left(n^{\gamma-1}\right)$\\
		\hline
	\end{tabular}
	}
}
\end{table*}

\begin{table*}[!ht]
	\caption{The probability of normal cellular flow $\operatorname{Pr}_{2}^{c}$}
	\renewcommand{\arraystretch}{1.3} 
	\label{tab_Pr2c}
	\centering{
		\scalebox{0.75}{
		\begin{tabular}{|c|c|c|c|}
			
			\hline
			\diagbox{$\alpha + \beta$}{$\alpha$} & $0 \leq \alpha<3$ & $\alpha=3$ & $\alpha>3$\\
			\hline
			\multicolumn{4}{|c|}{$\gamma>1$}\\
			\hline
			$0 \leq \alpha+\beta<3$ & $\Theta\left(r^{\beta}(n)-r^{3-\alpha}(n) L^{3-\alpha-\beta}\right)$ & -- & --\\
			\hline
			$a+\beta=3$ & $\Theta\left(r^{3-\alpha}(n) \ln \left(r^{-1}(n) L^{-1}\right)\right)$ & $\Theta\left(1-\log _{r^{-1}(n)} L\right)$ & --\\
			\hline
			$\alpha+\beta>3$&$\Theta\left(r^{\beta}(n)-r^{3-\alpha}(n) L^{3-\alpha-\beta}\right)$&$\Theta\left(\ln ^{-1}\left(r^{-1}(n)\right)\left(r^{\beta}(n)-L^{\beta}\right)\right)$&$\Theta\left(r^{\alpha+\beta-3}(n)-L^{3-\alpha-\beta}\right)$\\
			\hline
			\multicolumn{4}{|c|}{$0 \leq \gamma \leq 1$}\\
			\hline
			$0 \leq \alpha+\beta<3$ & $\Theta\left(n^{\gamma-1}\left(r^{\beta}(n)-r^{3-\alpha}(n) L^{3-\alpha-\beta}\right)\right)$ & -- & --\\
			\hline
			$\alpha+\beta=3$&$\Theta\left(n^{\gamma-1} r^{3-\alpha}(n) \ln \left(r^{-1}(n) L^{-1}\right)\right)$&$\Theta\left(n^{\gamma-1}\left(1-\log _{r^{-1}(n)} L\right)\right)$& --\\
			\hline
			$\alpha+\beta>3$ & $\Theta\left(n^{\gamma-1}\left(r^{\beta}(n)-r^{3-\alpha}(n) L^{3-\alpha-\beta}\right)\right)$ & $\Theta \left( {{n^{\gamma  - 1}}{{\ln }^{- 1}}\left( {{r^{- 1}}\left( n \right)} \right)\left( {{r^\beta }\left( n \right) - {L^{- \beta }}} \right)} \right)$ & $\Theta \left( {{n^{\gamma  - 1}}\left( {{r^{\alpha  + \beta  - 3}}\left( n \right) - {L^{3 - \alpha  - \beta }}} \right)} \right)$\\
			\hline
		\end{tabular}
			}
		} 
\end{table*}

Using the results of $\Pr _1^a$ in (\ref{eq_pr1a}), $\Pr _1^c$ in (\ref{eq_pr1c}), {$\Pr _2^a$ in Table \ref{tab_Pr2a}, and $\Pr _2^c$ in Table \ref{tab_Pr2c}},
the number of ad hoc flows $N_a$ and the number of cellular flows $N_c$ can be derived.
When $\gamma > 1$, the results of $N_a$ and $N_c$ are shown in (\ref{eq_46}) and (\ref{eq_47}).

\newcounter{mytempeqncnt}
\begin{figure*}[!ht]
	\normalsize
	\setcounter{mytempeqncnt}{\value{equation}}
	\setcounter{equation}{26}
	\begin{equation}\label{eq_46} \small
		{N_a} = \left\{{\begin{matrix}
				{\Theta \left( {n{r^3}\left( n \right){L^{3 - \beta }} + n{r^{3 - \alpha }}\left( n \right){L^{3 - \alpha  - \beta }}} \right)} \hfill & {0 \le \alpha  < 3,0 \le \alpha  + \beta  < 3} \hfill  \cr 
				{\Theta \left( {n{r^3}\left( n \right){L^{3 - \beta }} + n{r^{3 - \alpha }}\left( n \right)\ln L} \right)} \hfill & {0 \le \alpha  < 3,\alpha  + \beta  = 3} \hfill  \cr 
				{\Theta \left( {n{r^3}\left( n \right){L^{3 - \beta }} + n{r^{3 - \alpha }}\left( n \right)} \right)} \hfill & {0 \le \alpha  < 3,0 \le \beta  < 3,\alpha  + \beta  > 3} \hfill  \cr 
				{\Theta \left( {n{r^3}\left( n \right)\ln L + n{r^{3 - \alpha }}\left( n \right)} \right)} \hfill & {0 < \alpha  < 3,\beta  = 3} \hfill  \cr 
				{\Theta \left( {n{r^{3 - \alpha }}\left( n \right)} \right)} \hfill & {0 \le \alpha  < 3,\beta  > 3} \hfill  \cr 
				{\Theta \left( {n{r^3}\left( n \right){L^3} + n{{\log }_{{r^{- 1}}\left( n \right)}}L} \right)} \hfill & {\alpha  = 3,\beta  = 0} \hfill  \cr 
				{\Theta \left( {n{r^3}\left( n \right){L^{3 - \beta }} + n{{\ln }^{- 1}}\left( {{r^{- 1}}\left( n \right)} \right)} \right)} \hfill & {\alpha  = 3,0 < \beta  < 3} \hfill  \cr 
				{\Theta \left( {n{r^3}\left( n \right)\ln L + n{{\ln }^{- 1}}\left( {{r^{- 1}}\left( n \right)} \right)} \right)} \hfill & {\alpha  = 3,\beta  = 3} \hfill  \cr 
				{\Theta \left( {n{{\ln }^{- 1}}\left( {{r^{- 1}}\left( n \right)} \right)} \right)} \hfill & {\alpha  = 3,\beta  > 3} \hfill  \cr 
				{\Theta \left( n \right)} \hfill & {\alpha  > 3} \hfill  \cr 
		\end{matrix} } \right.
	\end{equation}
\begin{equation}\label{eq_47}\small
	N_{c}=\left\{\begin{array}{ll}
		\Theta\left(2 n r^{\beta}(n)-n r^{3}(n) L^{3-\beta}-n r^{3-\alpha}(n) L^{3-\alpha-\beta}\right) &
		0 \leq \alpha<3,0 \leq \beta<3, \alpha+\beta \neq 3 \\
		\Theta\left(n r^{\beta}(n)-n r^{3}(n) L^{3-\beta}+n r^{3-\alpha}(n) \ln \left(r^{-1}(n) L^{-1}\right)\right) & 0 \leq \alpha<3,0 \leq \beta \leq 3, \alpha+\beta=3 \\
		\Theta\left(n r^{3}(n) \ln \left(r^{-1}(n) L^{-1}\right)+n r^{\beta}(n)-n r^{3-\alpha}(n) L^{3-\alpha-\beta}\right) & 0<\alpha<3, \beta=3 \\
		\Theta\left(n r^{3}(n)+n r^{\beta}(n)-n r^{3-\alpha}(n) L^{3-\alpha-\beta}\right) & 0 \leq \alpha<3, \beta>3 \\
		\Theta\left(2 n-n r^{3}(n) L^{3}-n \log _{r^{-1}(n)} L\right) & \alpha=3, \beta=0 \\
		\Theta\left(n r^{\beta}(n)-n r^{3}(n) L^{3-\beta}+n \ln \left(r^{-1}(n) L^{-1}\right)\left(r^{\beta}(n)-L^{-\beta}\right)\right) & \alpha=3,0<\beta<3 \\
		\Theta\left(n \ln \left(r^{-1}(n) L^{-1}\right)\left(r^{3}(n)+r^{\beta}(n)-L^{-\beta}\right)\right) & \alpha=3, \beta=3 \\
		\Theta\left(n r^{3}(n)+n \ln \left(r^{-1}(n) L^{-1}\right)\left(r^{\beta}(n)-L^{-\beta}\right)\right) & \alpha=3, \beta>3 \\
		\Theta\left(n r^{\beta}(n)+n r^{\alpha+\beta-3}(n)-n r^{3}(n) L^{3-\beta}-n L^{3-\alpha-\beta}\right) & \alpha>3,0 \leq \beta<3 \\
		\Theta\left(n r^{3}(n) \ln \left(r^{-1}(n) L^{-1}\right)+n r^{\alpha}(n)-n r^{-\alpha}(n)\right) & \alpha>3, \beta=3 \\
		\Theta\left(n r^{3}(n)+n r^{\alpha+\beta-3}(n)-n L^{3-\alpha-\beta}\right) & \alpha>3, \beta>3
	\end{array}\right.
\end{equation}
\end{figure*}

Therefore, whether the flows are dominated by ad hoc flows or cellular flows is influenced by $L$ in the  routing strategy.
It is observed that when $\alpha > 3$, $N_a$ increases linearly with $n$.
Because when $\alpha$ increases, contact group members of source node gather around the source node, so that the number of hops needed for communication tends to be smaller than $L$, thus the number of ad hoc flows increases with $n$.
Besides, on account of $L = {\rm O}\left( {{r^{- 1}}\left( n \right)} \right)$ and $L = \Omega \left( 1 \right)$, we have ${N_a} + {N_c} \equiv \Theta \left( n \right)$.

\begin{Theo}
	The average number of hops of ad hoc flows passing through each cube, namely $E[F]$, is $\left( {n{{\Pr }^a}} \right) \times \left( {{1 \over {{{\Pr }^a}}}{E^\prime }\left[ X \right]} \right) \times {r^3}\left( n \right)$, where ${\Pr }^a$ is the probability for nodes transmitting with ad hoc mode.
\end{Theo}

\begin{proof}
	As mentioned in Section \ref{NC_of_adhoc}, $E\left[ F \right] = {N_a}E\left[ X \right]V$, where $V \equiv {r^3}\left( n \right)$.
	
	The total number of hops of all ad hoc flows is denoted by $X_{total}$. The number of hops of flow $i$ is denoted by $X_i$. Then, we have
	\begin{equation}\small
		E\left[ {{X_{total}}} \right] = E\left[ {\sum\limits_{i = 1}^{{N_a}} {{X_i}} } \right] = \sum\limits_{i = 1}^{{N_a}} {E\left[ {{X_i}} \right]}.
	\end{equation}
	
	Suppose that ${X_i}\left( {i \in \left\{{1,2,...{N_a}} \right\}} \right)$ are independent and identically distributed (i.i.d.), and $X$ has the same distribution with $X_i$.
	Then, we have $E\left[ {{X_{total}}} \right] = {N_a}E\left[ X \right]$.
	On the condition of unbiased estimation, $E\left[ {{X_{total}}} \right] = {N_a}E\left[ X \right]$. Therefore, for each cube, $E\left[ F \right] = {N_a}E\left[ X \right]V$, where $E\left[ X \right]$ is as \eqref{eq_new1}.

	\begin{equation}\label{eq_new1}\small
		\begin{aligned}
			E\left[ X \right] &= \sum\limits_{x = 1}^L {x\Pr \left( {X = x|\text{\emph{The flow is ad hoc flow}}} \right)}  \\&= \sum\limits_{x = 1}^L {x{{\Pr \left( {X = x} \right)} \over {{{\Pr }^a}}}}  \\&= {1 \over {{{\Pr }^a}}}\sum\limits_{x = 1}^L {x\Pr \left( {X = x} \right)}.
		\end{aligned}
	\end{equation}
	
	Suppose that $E^{\prime}[X]=\sum_{x=1}^{L} x \operatorname{Pr}(X=x)$, we have $E[X]=E^{\prime}[X] / \operatorname{Pr}^{a}$. We divide $E^{\prime}[X]$ into two parts according to the threshold $q_0$, i.e. $E^{\prime}[X]=E_{1}^{\prime}[X]+E_{2}^{\prime}[X]$, where 
	$E_{1}^{\prime}[X]$ represents the average number of hops of the flows starting from leader nodes, and $E_{2}^{\prime}[X]$ represents the average number of hops of the flows starting from normal nodes.
	
	Substituting (\ref{eq_18}) and (\ref{eq_19}) into $E_{1}^{\prime}[X]$ and $E_{2}^{\prime}[X]$, $E_{1}^{\prime}[X]$ is as follows.
	
	\begin{equation}\label{eq_50}\small
		E_{1}^{\prime}[X]=\left\{\begin{array}{ll}
			\Theta\left(r^{3}(n) L^{4-\beta}\right) & 0 \leq \beta<4 \\
			\Theta\left(r^{3}(n) \ln L\right) & \beta=4 \\
			\Theta\left(r^{3}(n)\right) & \beta>4
		\end{array}\right.
	\end{equation}

	When $\gamma > 1$, $E_{2}^{\prime}[X]$ is
	
	\begin{equation}\label{eq_51}\small
		\begin{small}
			E_{2}^{\prime}[X]=\left\{\begin{array}{ll}
			\Theta\left(r^{3-\alpha}(n) L^{4-\alpha-\beta}\right) & 0 \leq \alpha<3,0 \leq \alpha+\beta<4 \\
			\Theta\left(r^{3-\alpha}(n) \ln L\right) & 0 \leq \alpha<3, \alpha+\beta=4 \\
			\Theta\left(r^{3-\alpha}(n)\right) & 0 \leq \alpha<3, \alpha+\beta>4 \\
			\Theta\left(\ln ^{-1}\left(r^{-1}(n)\right) L^{1-\beta}\right) & \alpha=3,0 \leq \alpha+\beta<4 \\
			\Theta\left(\log _{r^{-1}(n)} L\right) & \alpha=3, \alpha+\beta=4 \\
			\Theta\left(\ln ^{-1}\left(r^{-1}(n)\right)\right) & \alpha=3, \alpha+\beta>4 \\
			\Theta\left(L^{4-\alpha-\beta}\right) & \alpha>3,0 \leq \alpha+\beta<4 \\
			\Theta(\ln L) & \alpha>3, \alpha+\beta=4 \\
			\Theta(1) & \alpha>3, \alpha+\beta>4
		\end{array}\right.
		\end{small}
	\end{equation}

\begin{figure*}[!ht]
	\normalsize
	\setcounter{mytempeqncnt}{\value{equation}}
		\begin{equation}\label{eq_52}\small
		E'\left[ X \right] = \left\{{\begin{matrix}
				{\Theta \left( {{r^3}\left( n \right){L^{4 - \beta }} + {r^{3 - \alpha }}\left( n \right){L^{4 - \alpha  - \beta }}} \right)} \hfill & {0 \le \alpha  < 3,0 \le \beta  < 4,0 \le \alpha  + \beta  < 4} \hfill  \cr 
				{\Theta \left( {{r^3}\left( n \right){L^{4 - \beta }} + {r^{3 - \alpha }}\left( n \right)\ln L} \right)} \hfill & {0 \le \alpha  < 3,0 \le \beta  \le 4,\alpha  + \beta  = 4} \hfill  \cr 
				{\Theta \left( {{r^3}\left( n \right){L^{4 - \beta }} + {r^{3 - \alpha }}\left( n \right)} \right)} \hfill & {0 \le \alpha  < 3,0 \le \beta  < 4,\alpha  + \beta  > 4} \hfill  \cr 
				{\Theta \left( {{r^3}\left( n \right)\ln L + {r^{3 - \alpha }}\left( n \right)} \right)} \hfill & {0 \le \alpha  < 3,\beta  = 4} \hfill  \cr 
				{\Theta \left( {{r^{3 - \alpha }}\left( n \right)} \right)} \hfill & {0 \le \alpha  < 3,\beta  > 4} \hfill  \cr 
				{\Theta \left( {{r^3}\left( n \right){L^{4 - \beta }} + {{\ln }^{- 1}}\left( {{r^{- 1}}\left( n \right)} \right){L^{1 - \beta }}} \right)} \hfill & {\alpha  = 3,0 \le \beta  < 1} \hfill  \cr 
				{\Theta \left( {{r^3}\left( n \right){L^3} + {{\log }_{{r^{- 1}}\left( n \right)}}L} \right)} \hfill & {\alpha  = 3,\beta  = 1} \hfill  \cr 
				{\Theta \left( {{r^3}\left( n \right){L^{4 - \beta }} + {{\ln }^{- 1}}\left( {{r^{- 1}}\left( n \right)} \right)} \right)} \hfill & {\alpha  = 3,1 < \beta  < 4} \hfill  \cr 
				{\Theta \left( {{r^3}\left( n \right)\ln L + {{\log }_{{r^{- 1}}\left( n \right)}}L} \right)} \hfill & {\alpha  = 3,\beta  = 4} \hfill  \cr 
				{\Theta \left( {{{\ln }^{- 1}}\left( {{r^{- 1}}\left( n \right)} \right)} \right)} \hfill & {\alpha  = 3,\beta  > 4} \hfill  \cr 
				{\Theta \left( {{r^3}\left( n \right){L^{4 - \beta }} + {L^{4 - \alpha  - \beta }}} \right)} \hfill & {\alpha  > 3,0 \le \beta  < 4,0 \le \alpha  + \beta  < 4} \hfill  \cr 
				{\Theta \left( {{r^3}\left( n \right){L^{4 - \beta }} + \ln L} \right)} \hfill & {\alpha  > 3,0 \le \beta  < 4,\alpha  + \beta  = 4} \hfill  \cr 
				{\Theta \left( {{r^3}\left( n \right){L^{4 - \beta }}} \right)} \hfill & {\alpha  > 3,0 \le \beta  < 4,\alpha  + \beta  > 4} \hfill  \cr 
				{\Theta \left( {{r^3}\left( n \right)\ln L} \right)} \hfill & {\alpha  > 3,\beta  = 4} \hfill  \cr 
				{\Theta \left( {{r^3}\left( n \right)} \right)} \hfill & {\alpha  > 3,\beta  > 4} \hfill  \cr
		\end{matrix}} \right.
	\end{equation}
\end{figure*}


	When $0 \le \gamma \le 1$, $E_{2}^{\prime}[X]$ is $n^{\gamma-1}$ times greater than that in (\ref{eq_51}).	
	(\ref{eq_50}) and (\ref{eq_51}) show that if the source nodes are leader nodes, the average number of hops of ad hoc flows is only related to $\beta$. If the source nodes are normal nodes, $\alpha$ and $\beta$ will jointly influence the average number of hops of the ad hoc flows and cellular flows of the source nodes. 
	This is due to the fact that the leader nodes connect more members, which counteracts the influence of the concentration factor $\alpha$ when $n$ tends to infinity. 

	Specifically, the distance between source node and destination node will decrease when $\alpha$ or $\beta$ increases.
	For leader nodes, when $0\le \beta \le 4$, the average number of hops of ad hoc flows increases when $L$ increases, and decreases when $\beta$ increases. When $\beta$ is large, the destination node tends to be closer to the source node. Therefore, when $\beta>4$, the trend of $E_{1}^{\prime}[X]$ has nothing to do with $L$.
	For normal nodes, $E_{2}^{\prime}[X]$ are influenced by $\alpha$ and $\beta$ simultaneously. When $0 \le \alpha <3$ and $0\le \alpha+\beta <4$, the average number of hops of the ad hoc flows increases when $L$ increases, and decreases when $\alpha$ or $\beta$ increases. When $\alpha >3$ or $\alpha + \beta>4$, $L$ has no relationship with the number of average number of hops.
	
	Since $\gamma >2$ in the actual network \cite{wiki1}\cite{wiki2}, we can derive the result of $E^{\prime}[X]$ as (\ref{eq_52}).

Finally, we have the following result.
\begin{equation}\small
	\begin{aligned}
		E\left[ F \right] &= {N_a}E\left[ X \right]V = \left( {n{{\Pr }^a}} \right) \times \left( {{1 \over {{{\Pr }^a}}}{E^\prime }\left[ X \right]} \right) \times {r^3}\left( n \right) \\&= \log \left( n \right){E^\prime }\left[ X \right].
	\end{aligned}
\end{equation}
\end{proof}

\subsection{Network throughput}
According to (\ref{eq_16}), the per-node throughput of ad hoc mode $\lambda_{a}^{n}$ is 
\begin{equation}\small
	\lambda_{a}^{n} \equiv \Theta\left(\frac{W_{a}}{E[F]}\right)=\Theta\left(\frac{W_{a}}{\log (n) E^{\prime}[X]}\right).
\end{equation}

Note that if $E[F]=O(1)$, $\lambda_{a}^{n}$ equals to $\Theta\left(W_{a}\right)$, because the average throughput of the ad hoc flows is smaller than $W_a$.

The relationship between $E^{\prime}[X]$ and $L$ under different ranges of $\alpha$ and $\beta$ is analyzed as follows.
Note that since $\gamma > 2$ in the actual network \cite{wiki1}\cite{wiki2}, only the results of $\gamma >1$ is considered in terms of the number of the ad hoc flows $N_a$.

{ In order to better understand the piecewise of throughput as follows, recall that the probability that a node is selected as a member of contact group is proportional to ${d_i}^{- \alpha }$, the probability that a contact group member will be communicated in a certain time slot is proportional to ${d_i}^{- \beta }$, and the number of members in the contact group is proportional to ${q^{- \gamma }}$. Therefore, the probability that a node will be communicated in a certain time slot is proportional to ${d_i}^{- (\alpha  + \beta )}$, i.e., $(\alpha  + \beta )$ reveals the communication activity level of the network.}
\subsubsection{$0 \le \alpha <3$, $0\le\beta<3$ and $0\le\alpha+\beta<3$}\label{sec_a}
When $L=\Omega\left(r^{-1}(n)\right)$, $E^{\prime}[X]$ is dominated by $E_{1}^{\prime}[X]$. When $L=\mathrm{O}\left(r^{-1}(n)\right)$, $E^{\prime}[X]$ is dominated by $E_{2}^{\prime}[X]$.
Considering that in the unit cube of the communication model in Section \ref{sec_interference_model}, there is always $L=\mathrm{O}\left(r^{-1}(n)\right)$.
Therefore, $E^{\prime}[X]$ is always dominated by $E_{2}^{\prime}[X]$ in this case.
So the per-node throughput of ad hoc mode is as follows.
\begin{equation}\label{eq_55}\small
	\begin{small}
		\lambda_{a}^{n} \equiv\left\{\begin{array}{ll}
			\Theta\left(\frac{W_{a}}{\log (n) r^{3-\alpha}(n) L^{4-\alpha-\beta}}\right)  \\\qquad \qquad \qquad  L=\Omega\left(\left(\log ^{-1}(n) r^{\alpha-3}(n)\right)^{\frac{1}{4-\alpha-\beta}}\right) \\
			\Theta\left(W_{a}\right)  \\\qquad \qquad \qquad  L=\mathrm{O}\left(\left(\log ^{-1}(n) r^{\alpha-3}(n)\right)^{\frac{1}{4-\alpha-\beta}}\right)
		\end{array}\right.
	\end{small}
\end{equation}

$\bullet$ When $L=\Omega\left(\left(\log ^{-1}(n) r^{\alpha-3}(n)\right)^{1 / 4-\alpha-\beta}\right)$, the network throughput of ad hoc mode is
\begin{equation}\label{eq_56}\small
	\begin{aligned}
		\lambda_{a} \equiv N_{a} \lambda_{a}^{n} &=\Theta\left(n r^{3}(n) L^{3-\beta}+n r^{3-\alpha}(n) L^{3-\alpha-\beta}\right) \\ &\quad\cdot \Theta\left(\frac{W_{a}}{\log (n) r^{3-\alpha}(n) L^{4-\alpha-\beta}}\right) \\
		&=\Theta\left(\frac{n^{1-\frac{\alpha}{3}} L^{\alpha-1} W_{a}}{\log ^{1-\frac{\alpha}{3}}(n)}+\frac{n W_{a}}{\log (n) L}\right).
	\end{aligned}
\end{equation}

$\bullet$ When $L=\mathrm{O}\left(\left(\log ^{-1}(n) r^{\alpha-3}(n)\right)^{1 / 4-\alpha-\beta}\right)$, the network throughput of ad hoc mode is

\begin{equation}\label{eq_57}\small
	\begin{aligned}
		\lambda_{a} &\equiv N_{a} \lambda_{a}^{n}\\&=\Theta\left(\log (n) L^{3-\beta} W_{a}+n^{\frac{\alpha}{3}} \log ^{1-\frac{\alpha}{3}}(n) L^{3-\alpha-\beta} W_{a}\right).
	\end{aligned}
\end{equation}

According to (\ref{eq_56}) and (\ref{eq_57}), when $L=\Theta\left(\left(\log ^{-1}(n) r^{\alpha-3}(n)\right)^{1 / 4-\alpha-\beta}\right)$, the network throughput of ad hoc mode is dominant, which is
\begin{equation}\small\label{eq39}
	\lambda_{a }=\Theta\left(\log (n) L^{3-\beta} W_{a}+n^{\frac{\alpha}{3}} \log ^{1-\frac{\alpha}{3}}(n) L^{3-\alpha-\beta} W_{a}\right).
\end{equation}

\subsubsection{$0 \le \alpha <3$, $0 \le \beta <3$ and $3 < \alpha + \beta <4$}\label{sec_b}
In this case, $E^{\prime}[X]$ is dominated by $E_{2}^{\prime}[X]$, where $L=\mathrm{O}\left(r^{-1}(n)\right)$. Therefore, the per-node throughput of ad hoc mode is (\ref{eq_55}).

$\bullet$ When $L=\Omega\left(\left(\log ^{-1}(n) r^{\alpha-3}(n)\right)^{1 / 4-\alpha-\beta}\right)$, the network throughput of ad hoc mode is 
\begin{equation}\small
	\begin{aligned}
		\lambda_{a} \equiv N_{a} \lambda_{a}^{n} &=\Theta\left(n r^{3}(n) L^{3-\beta}+n r^{3-\alpha}(n)\right)\\&\quad \cdot \Theta\left(\frac{W_{a}}{\log (n) r^{3-\alpha}(n) L^{4-\alpha-\beta}}\right) \\
		&=\Theta\left(\frac{n^{1-\frac{\alpha}{3}} L^{\alpha-1} W_{a}}{\log ^{1-\frac{\alpha}{3}}(n)}+\frac{n W_{a}}{\log (n) L^{4-\alpha-\beta}}\right).
	\end{aligned}
\end{equation}

$\bullet$ When $L=\mathrm{O}\left(\left(\log ^{-1}(n) r^{\alpha-3}(n)\right)^{1 / 4-\alpha-\beta}\right)$, the network
throughput of ad hoc mode is
\begin{equation}\small
	\lambda_{a} \equiv N_{a} \lambda_{a}^{n}=\Theta\left(\log (n) L^{3-\beta} W_{a}+n^{\frac{\alpha}{3}} \log ^{1-\frac{\alpha}{3}}(n) W_{a}\right).
\end{equation}

When $L=\Theta\left(\left(\log ^{-1}(n) r^{\alpha-3}(n)\right)^{1 / 4-\alpha-\beta}\right)$, the network throughput of ad hoc mode is dominant, which is
\begin{equation}\small\label{eq42}
	\lambda_{a }=\Theta\left(\log (n) L^{3-\beta} W_{a}+n^{\frac{\alpha}{3}} \log ^{1-\frac{\alpha}{3}}(n) W_{a}\right).
\end{equation}

\subsubsection{$0 \le \alpha <3$, $0 \le \beta <3$ and $\alpha + \beta >4$}\label{sec_c}
In this case, $E^{\prime}[X]$ is dominated by $E_{1}^{\prime}[X]$ when $L=\Omega\left(r^{\alpha / \beta-4}(n)\right)$, and $E^{\prime}[X]$ is dominated by $E_{2}^{\prime}[X]$ when $L=\mathrm{O}\left(r^{\alpha / \beta-4}(n)\right)$.

$\bullet$ When normal nodes are dominant, $E[F]=\log (n) E[X]=O(1)$. Hence, $\lambda_{a}^{n} \equiv \Theta\left(W_{a}\right)$.
Note that when $\alpha \equiv \Omega\left(3+\log _{r(n)} \log (n)\right)$, we have $r^{\alpha / (\beta-4)}(n) \equiv \Omega\left(\left(\log ^{-1}(n) r^{-3}(n)\right)^{1 / 4-\beta}\right)$.
However, when $n \rightarrow \infty$, we have $\alpha \equiv \mathrm{O}\left(3+\log _{r(n)} \log (n)\right)$. 
In this case, $r^{\alpha / \beta-4}(n) \equiv \mathrm{O}\left(\left(\log ^{-1}(n) r^{-3}(n)\right)^{1 / 4-\beta}\right)$. Thus, the per-node throughput of ad hoc mode is 
\begin{equation}\small
	\begin{small}
			\lambda_{a}^{n} \equiv\left\{\begin{array}{cl}
			\Theta\left(\frac{W_{a}}{\log (n) r^{3}(n) L^{4-\beta}}\right) & L=\Omega\left(\left(\log ^{-1}(n) r^{-3}(n)\right)^{\frac{1}{4-\beta}}\right) \\
			\Theta\left(W_{a}\right) & L=\mathrm{O}\left(\left(\log ^{-1}(n) r^{-3}(n)\right)^{\frac{1}{4-\beta}}\right)
		\end{array}\right.
	\end{small}
\end{equation}

$\bullet$ When $L=\Omega\left(\left(\log ^{-1}(n) r^{-3}(n)\right)^{1 / 4-\beta}\right)$, the network throughput of ad hoc mode is 
\begin{equation}\label{eq_63}\small
	\begin{aligned}
		\lambda_{a} \equiv N_{a} \lambda_{a}^{n} &=\Theta\left(n r^{3}(n) L^{3-\beta}+n r^{3-\alpha}(n)\right) \\&\quad\cdot \Theta\left(\frac{W_{a}}{\log (n) r^{3}(n) L^{4-\beta}}\right) \\
		&=\Theta\left(\frac{n W_{a}}{\log (n) L}+\frac{n^{1+\frac{\alpha}{3}} W_{a}}{\log \sqrt{3}^{1+\frac{\alpha}{3}}(n) L^{4-\beta}}\right).
	\end{aligned}
\end{equation}

$\bullet$ When $L=\mathrm{O}\left(\left(\log ^{-1}(n) r^{\alpha-3}(n)\right)^{1 / 4-\alpha-\beta}\right)$, the network throughput of ad hoc mode is
\begin{equation}\label{eq_64}\small
	\lambda_{a} \equiv N_{a} \lambda_{a}^{n}=\Theta\left(\log (n) L^{3-\beta} W_{a}+n^{\frac{\alpha}{3}} \log ^{1-\frac{\alpha}{3}}(n) W_{a}\right).
\end{equation}

According to (\ref{eq_63}) and (\ref{eq_64}), when $L=\Theta\left(\left(\log ^{-1}(n) r^{\alpha-3}(n)\right)^{1 / 4-\alpha-\beta}\right)$,  the network throughput of ad hoc mode is dominant, which is
\begin{equation}\small\label{eq46}
	\lambda_{a }=\Theta\left(\log (n) L^{3-\beta} W_{a}+n^{\frac{\alpha}{3}} \log ^{1-\frac{\alpha}{3}}(n) W_{a}\right).
\end{equation}

\subsubsection{$0\le \alpha <3$, $\beta >3$ and $3<\alpha +\beta <4$}\label{sec_d}

In this case, $E^{\prime}[X]$ is equal to that of subsection \emph{1)} and \emph{2)}, in which $L=\mathrm{O}\left(r^{-1}(n)\right)$, and $E^{\prime}[X]$ is dominated by $E_{2}^{\prime}[X]$.

$\bullet$ When $L=\Omega\left(\left(\log ^{-1}(n) r^{\alpha-3}(n)\right)^{1 / 4-\alpha-\beta}\right)$, the network throughput of ad hoc mode is 
\begin{equation}\small
	\begin{aligned}
		\lambda_{a} \equiv N_{a} \lambda_{a}^{n} &=\Theta\left(n r^{3-\alpha}(n)\right) \cdot \Theta\left(\frac{W_{a}}{\log (n) r^{3-\alpha}(n) L^{4-\alpha-\beta}}\right) \\
		&=\Theta\left(\frac{n W_{a}}{\log (n) L^{4-\alpha-\beta}}\right).
	\end{aligned}
\end{equation}

$\bullet$ When $L=\Omega\left(\left(\log ^{-1}(n) r^{\alpha-3}(n)\right)^{1 / 4-\alpha-\beta}\right)$, the network throughput is
\begin{equation}\small
	\lambda_{a} \equiv N_{a} \lambda_{a}^{n}=\Theta\left(n^{\frac{\alpha}{3}} \log ^{1-\frac{\alpha}{3}}(n) W_{a}\right).
\end{equation}

Therefore, when $L=\Theta\left(\left(\log ^{-1}(n) r^{\alpha-3}(n)\right)^{1 / 4-\alpha-\beta}\right)$, the network throughput of ad hoc mode is dominant, which is 
\begin{equation}\small\label{eq49}
	\lambda_{a}=\Theta\left(n^{\frac{\alpha}{3}} \log ^{1-\frac{\alpha}{3}}(n) W_{a}\right),
\end{equation}
which reveals that when $\beta$ is large enough, the destination nodes will gather around the source nodes, which make the network throughput independent of $L$, and the number of ad hoc transmission depends on the value of $\alpha$ in a certain range.

\subsubsection{$\alpha >3$, $0 \le \beta <4$ and $3 < \alpha +\beta<4$}\label{sec_e}
In this case, $E^{\prime}[X]$ is dominated by $E_{1}^{\prime}[X]$ when $L=\Omega\left(r^{-3 / \alpha}(n)\right)$, and $E^{\prime}[X]$ is dominated by $E_{2}^{\prime}[X]$ when $L=\mathrm{O}\left(r^{-3 / \alpha}(n)\right)$. When $L=\mathrm{O}\left(r^{-3 / \alpha}(n)\right)$, there is always $r^{-3 / \alpha}(n)=\Omega\left(\left(\log ^{-1}(n) r^{-3}(n)\right)^{1 / 4-\beta}\right)$. Therefore, the per-node throughput of ad hoc mode is as (\ref{eq_69}).
\begin{figure*}[!ht]
	\normalsize
	\setcounter{mytempeqncnt}{\value{equation}}
	\begin{equation}\label{eq_69}\small
		\lambda_{a}^{n} \equiv\left\{\begin{array}{lc}
			\Theta\left(\frac{W_{a}}{\log (n) r^{3}(n) L^{4-\beta}}\right) & L=\Omega\left(\left(\log ^{-1}(n) r^{-3}(n)\right)^{\frac{1}{4-\beta}}\right) \\
			\Theta\left(\frac{W_{a}}{\log (n) L^{4-\alpha-\beta}}\right) & L=\mathrm{O}\left(\left(\log ^{-1}(n) r^{-3}(n)\right)^{\frac{1}{4-\beta}}\right) \text {and } L=\Omega\left(\log ^{-\frac{1}{4-\alpha-\beta}}(n)\right) \\
			\Theta\left(W_{a}\right) & L=\mathrm{O}\left(\log ^{-\frac{1}{4-\alpha-\beta}}(n)\right)
		\end{array}\right.
	\end{equation}
\end{figure*}

$\bullet$ When $L=\Omega\left(\left(\log ^{-1}(n) r^{-3}(n)\right)^{1 / (4-\beta)}\right)$, the network throughput of ad hoc mode is
\begin{equation}\small
	\begin{aligned}
		\lambda_{a} \equiv N_{a} \lambda_{a}^{n}&=\Theta(n) \cdot \Theta\left(\frac{W_{a}}{\log (n) r^{3}(n) L^{4-\beta}}\right)\\&=\Theta\left(\frac{n^{2} W_{a}}{\log ^{2}(n) L^{4-\beta}}\right).
	\end{aligned}
\end{equation}

$\bullet$ When $L=\mathrm{O}\left(\left(\log ^{-1}(n) r^{-3}(n)\right)^{1 / (4-\beta)}\right)$, and $L=\Omega\left(\log ^{-1 / (4-\alpha-\beta)}(n)\right)$, 
\begin{equation}\small
	\lambda_{a} \equiv N_{a} \lambda_{a}^{n}=\Theta\left(\frac{n W_{a}}{\log (n) L^{4-\alpha-\beta}}\right).
\end{equation}

$\bullet$ When $L=\mathrm{O}\left(\log ^{-1 / (4-\alpha-\beta)}(n)\right)$, the network throughput of ad hoc mode is
\begin{equation}\small
	\lambda_{a} \equiv N_{a} \lambda_{a}^{n}=\Theta\left(n W_{a}\right).
\end{equation}

Therefore, when $L=\Theta\left(\log ^{-1 / 4-\alpha-\beta}(n)\right)$, the network throughput of ad hoc mode is dominant, which is
\begin{equation}\small\label{eq54}
	\lambda_{a}=\Theta\left(n W_{a}\right).
\end{equation}

The result also shows that in this range of parameters, when $L$ breaks the boundary of $\Theta\left(\log ^{-1 / 4-\alpha-\beta}(n)\right)$, the network throughput will have nothing to do with $L$, and all the nodes will be in ad hoc mode.

\subsubsection{$\alpha >3$, $0 \le \beta <4$ and $\alpha + \beta >4$}
In this case, $E^{\prime}[X]$ is dominated by $E_{1}^{\prime}[X]$. The per-node throughput of ad hoc mode is 
\begin{equation}\small
	\begin{small}
		\lambda_{a}^{n} \equiv\left\{\begin{array}{cl}
			\Theta\left(\frac{W_{a}}{\log (n) r^{3}(n) L^{4-\beta}}\right) & L=\Omega\left(\left(\log ^{-1}(n) r^{-3}(n)\right)^{\frac{1}{4-\beta}}\right) \\
			\Theta\left(W_{a}\right) & L=\mathrm{O}\left(\left(\log ^{-1}(n) r^{-3}(n)\right)^{\frac{1}{4-\beta}}\right)
		\end{array}\right.
	\end{small}
\end{equation}

$\bullet$ When $L=\Omega\left(\left(\log ^{-1}(n) r^{-3}(n)\right)^{1 / 4-\beta}\right)$, the network throughput of ad hoc mode is 
\begin{equation}\small
	\begin{aligned}
		\lambda_{a} \equiv N_{a} \lambda_{a}^{n}&=\Theta(n) \cdot \Theta\left(\frac{W_{a}}{\log (n) r^{3}(n) L^{4-\beta}}\right)\\&=\Theta\left(\frac{n^{2} W_{a}}{\log ^{2}(n) L^{4-\beta}}\right).
	\end{aligned}
\end{equation}

$\bullet$ When $L=\mathrm{O}\left(\left(\log ^{-1}(n) r^{-3}(n)\right)^{1 / 4-\beta}\right)$, the network throughput of ad hoc mode is 
\begin{equation}\small\label{eq57}
	\lambda_{a} \equiv N_{a} \lambda_{a}^{n}=\Theta\left(n W_{a}\right).
\end{equation}

Therefore, when $L=\Theta\left(\left(\log ^{-1}(n) r^{-3}(n)\right)^{1 / 4-\beta}\right)$, the network throughput of ad hoc mode is dominant, which is 
\begin{equation}\small
	\lambda_{a}=\Theta\left(n W_{a}\right).
\end{equation}

The result shows that the increase of $\alpha$ influences the distribution of the destination nodes. When $L$ breaks the boundary of $\Theta\left(\left(\log ^{-1}(n) r^{-3}(n)\right)^{1 / 4-\beta}\right)$, the network throughput will no longer be related to $L$, and all the nodes will be in ad hoc mode.

\subsubsection{$\alpha >3$ and $\beta>4$}
In this case, $E^{\prime}[X]=\Theta\left(r^{3}(n)\right)$, and $E[F]=\log (n) E^{\prime}[X]=\mathrm{O}(1)$ when $n \rightarrow \infty$. 
Therefore, the per-node throughput of ad hoc mode is 
\begin{equation}\small
	\lambda_{a}^{n} \equiv \Theta\left(W_{a}\right).
\end{equation}

When $\alpha$ and $\beta$ are large, the destination nodes are close to the source nodes.
As a result, the number of hops are always smaller than $L$, and all of the nodes transmit in ad hoc mode.
In this case, the network throughput has nothing to do with $L$. The network throughput of ad hoc mode is
\begin{equation}\small\label{eq60}
	\lambda_{a} \equiv N_{a} \lambda_{a}^{n}=\Theta(n) \cdot \Theta\left(W_{a}\right)=\Theta\left(n W_{a}\right).
\end{equation}


{ In conclusion, the aggregation of destination nodes (i.e. the value of $\alpha + \beta$) is the main factor affecting the network throughput. When the distribution of destination nodes is sparse, the network throughput has complex relationships with hop threshold $L$, the number of nodes $n$, and the bandwidth $W_a$, as shown in \eqref{eq39}\eqref{eq42}\eqref{eq46}. When the destination nodes gather around the source nodes, the flows in the network are generally ad hoc flows. The hop threshold $L$ has limited impact on the network throughput, and the throughput is only positively related to the number of nodes $n$ and bandwidth $W_a$, as shown in \eqref{eq49}\eqref{eq54}\eqref{eq57}\eqref{eq60}.
Furthermore, as shown in \eqref{eq54}\eqref{eq57}\eqref{eq60}, when the destination nodes have strong aggregation to the source nodes, the network throughput will be in direct proportion to the product of $n$ and $W_a$.}

\section{Numerical Results and Analysis}
{ In this section, the theoretical results in Section III are verified by numerical results, and the relationship between parameters and network throughput is analyzed through the numerical results. Besides, the network throughput of 100 to 10000 UAVs is simulated by MATLAB to verify the rationality of the theoretical results.}

\begin{figure*}[!ht]
	\centering
	\includegraphics[width=0.99\textwidth]{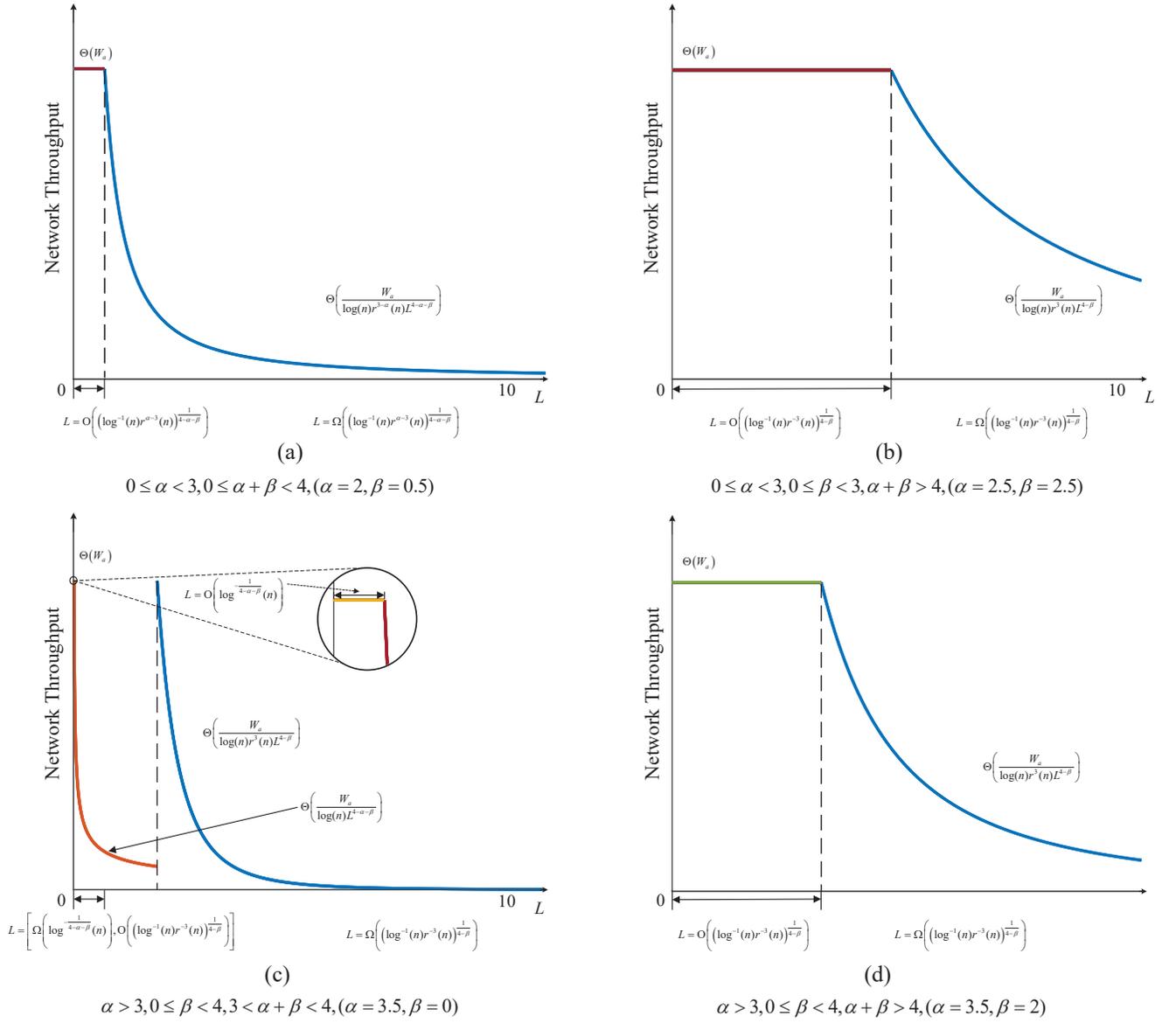}
	\caption{The relationship between $L$ and throughput of ad hoc network}
	\label{Combine_4}
\end{figure*} 
According to the theoretical results derived above, Fig. \ref{Combine_4} shows the relation between the threshold $L$ and the throughput with different values of $\alpha$ and $\beta$.
We consider four typical conditions, where $n = 100$, $W_a = 1$, and the optimal $L$ are identified.

In Fig. \ref{Combine_4}(a), where $0\le \alpha <3$ and $0 \le \alpha+\beta <4$, the optimal $L$ is relatively small.
Since the values of $\alpha$ and $\beta$ are small, there are more long-distance flows.
Besides, because the leader nodes have large contact groups, these long-distance flows are more likely to be sent by leader nodes and their hops are more likely to be larger than $L$, which means that most of the long-distance flows are cellular flows.
In this case, the throughput of the ad hoc network is dominated by flows of normal nodes.

In Fig. \ref{Combine_4}(b), where $0 \le \alpha<3$, $0 \le \beta <3$ and $\alpha + \beta >4$, the optimal $L$ is relatively large.
At this time, because $\alpha$ and $\beta$ are small, there are still many long-distance flows from the leader nodes.
However, since $L$ increases, the number of long-distance flows with ad hoc mode increases, and the throughput tends to be dominated by leader nodes.

In Fig. \ref{Combine_4}(c), where $\alpha>3$, $0\le\beta<4$ and $3<\alpha+\beta<4$, there are two optimal $L$, which shows that the value of $L$ determines the type of nodes that dominate the throughput.
When $L=\Omega\left(\left(\log ^{-1}(n) r^{-3}(n)\right)^{1 / (4-\beta)}\right)$, the throughput is dominated by leader nodes. When $L=\mathrm{O}\left(\left(\log ^{-1}(n) r^{-3}(n)\right)^{1 / (4-\beta)}\right)$ and $L=\Omega\left(\log ^{-1 / (4-\alpha-\beta)}(n)\right)$, the throughput is dominated by normal nodes.
Because $\alpha$ is large, the aggregation of contact groups of source nodes is high.
However, $\beta$ and $\alpha + \beta$ are still small. Thus, it is still possible for source nodes to communicate with contact group nodes with a long distance.
Due to the large number of contact group members of leader nodes, it is more likely that such long-distance flows will be sent by leader nodes.
It can be explained that when $L$ is large, more long-distance flows sent by leader nodes are transmitted by ad hoc mode, which dominates the network throughput.
When $L$ is small, leader nodes prefer cellular mode, so that the throughput of ad hoc flows is dominated by normal nodes.

In Fig. \ref{Combine_4}(d), as $\alpha$ increases, the aggregation of the contact groups is further improved. 
It is more likely that the number of hops of long-distance flows is less than $L$, so that the throughput of ad hoc flows is dominated by leader nodes again.

The theoretical results above show that there is an optimal value of $L$ to maximize the average throughput of UAV network, which is of great significance to the design of hybrid UAV network.
For example, when the number of UAVs and the capabilities of UAVs are determined, the parameters $\alpha$, $\beta$ and $\gamma$ can be determined by analyzing the routing table and the topological relation of UAV  network.
With such parameters, we can determine the value of routing strategy $L$ in hybrid UAV network to maximize the throughput of UAV network.

{ In order to verify the theoretical results through simulation, firstly, the Bat Algorithm (BA) algorithm is applied to generate a scale-free network which has the same setting as the models in Section II. Taking 100 nodes as an example, the contact groups and communication relationships of each node are shown in Fig. \ref{Add_pic_1} and Fig. \ref{Add_pic_2}. }

{ Fig. \ref{Add_pic_1} is the contact group selection when $n=100$, $\alpha = 1$, $\beta = 0.5$, and $\gamma = 2$. The UAVs are randomly distributed in the 3D space. The selection of the contact group members of the source nodes follows the power-law distribution with parameter $\alpha$. In the simulation of Fig. \ref{Add_pic_1}, the threshold $q_0$ is set to be 17.33. The source node of Fig. \ref{Add_pic_1}(a) is a leader node with 25 social group members. The source node of Fig. \ref{Add_pic_1}(b) is a normal node with 5 social group members.}

\begin{figure*}[!ht]
	\centering
	\includegraphics[width=0.99\textwidth]{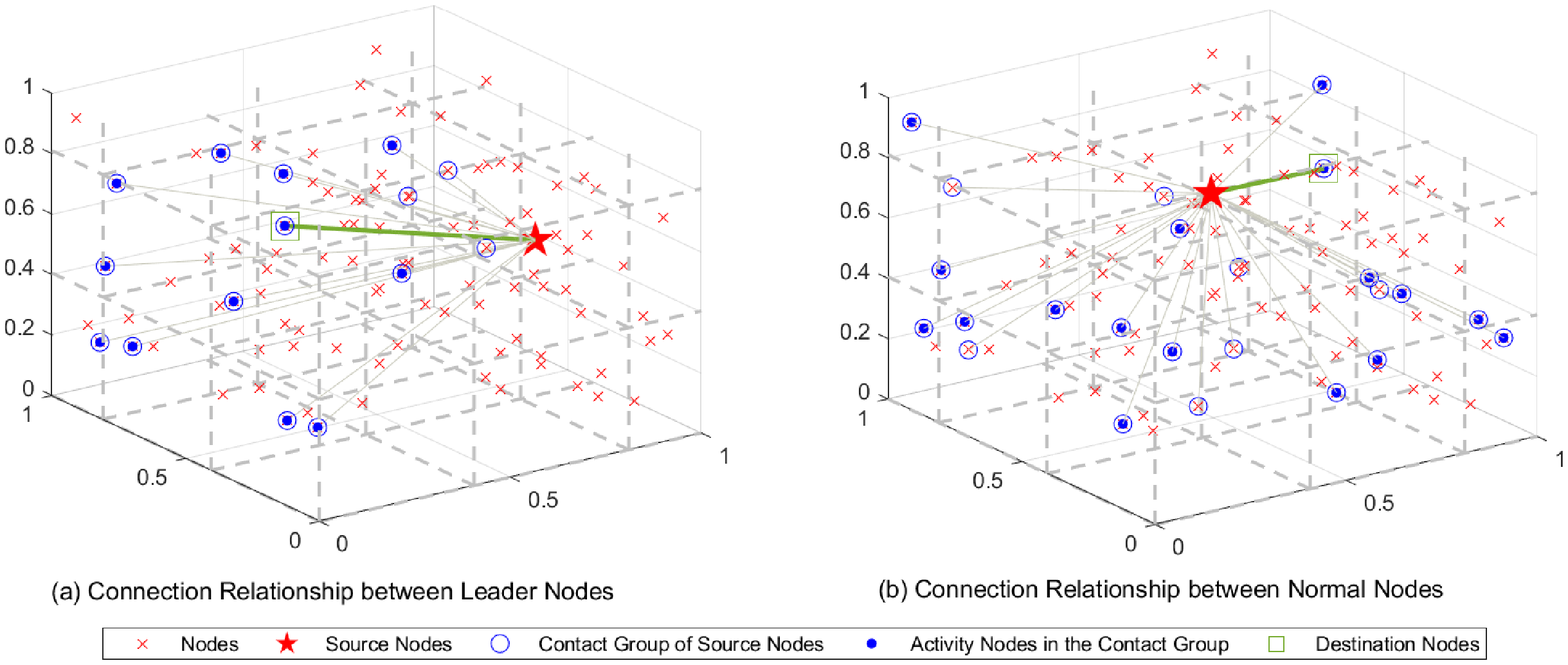}
	\caption{The contact groups of leader nodes and normal nodes}
	\label{Add_pic_1}
\end{figure*} 

{ Fig. \ref{Add_pic_2} illustrates the contact group selection related to parameter $\alpha$, the communication selection in the contact group related to parameter $\beta$, and the final communication relationship. The three sub-figures in Fig. \ref{Add_pic_2} are all directed graphs. The evolution process is revealed from the first sub-figure to the last sub-figure.}
\begin{figure*}[!ht]
	\centering
	\subfigure[Contact Group Selection]{\includegraphics[width=0.45\textwidth]{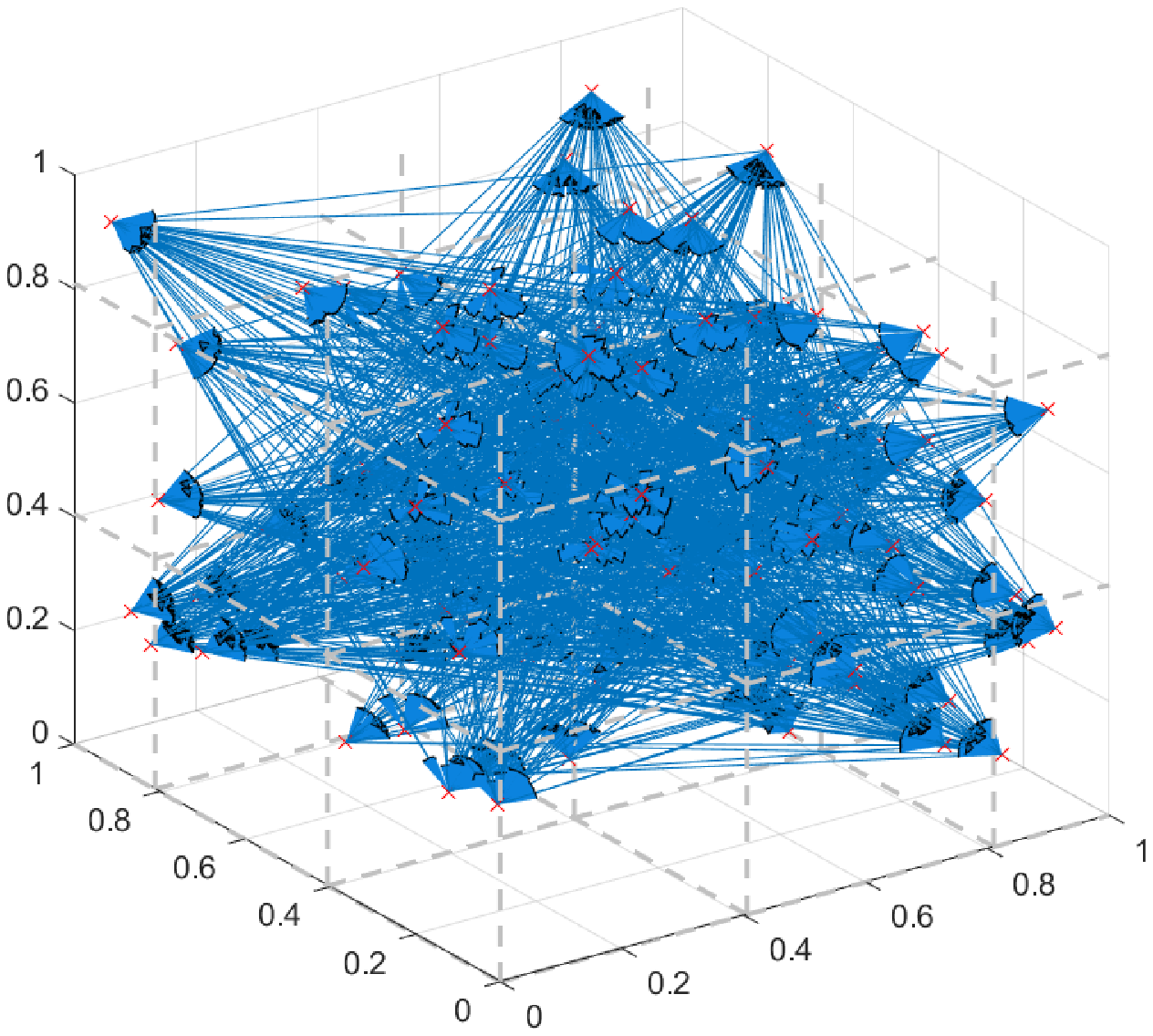}\label{6a}}
	\quad
	\subfigure[Communication Selection in the Contact Group]{\includegraphics[width=0.45\textwidth]{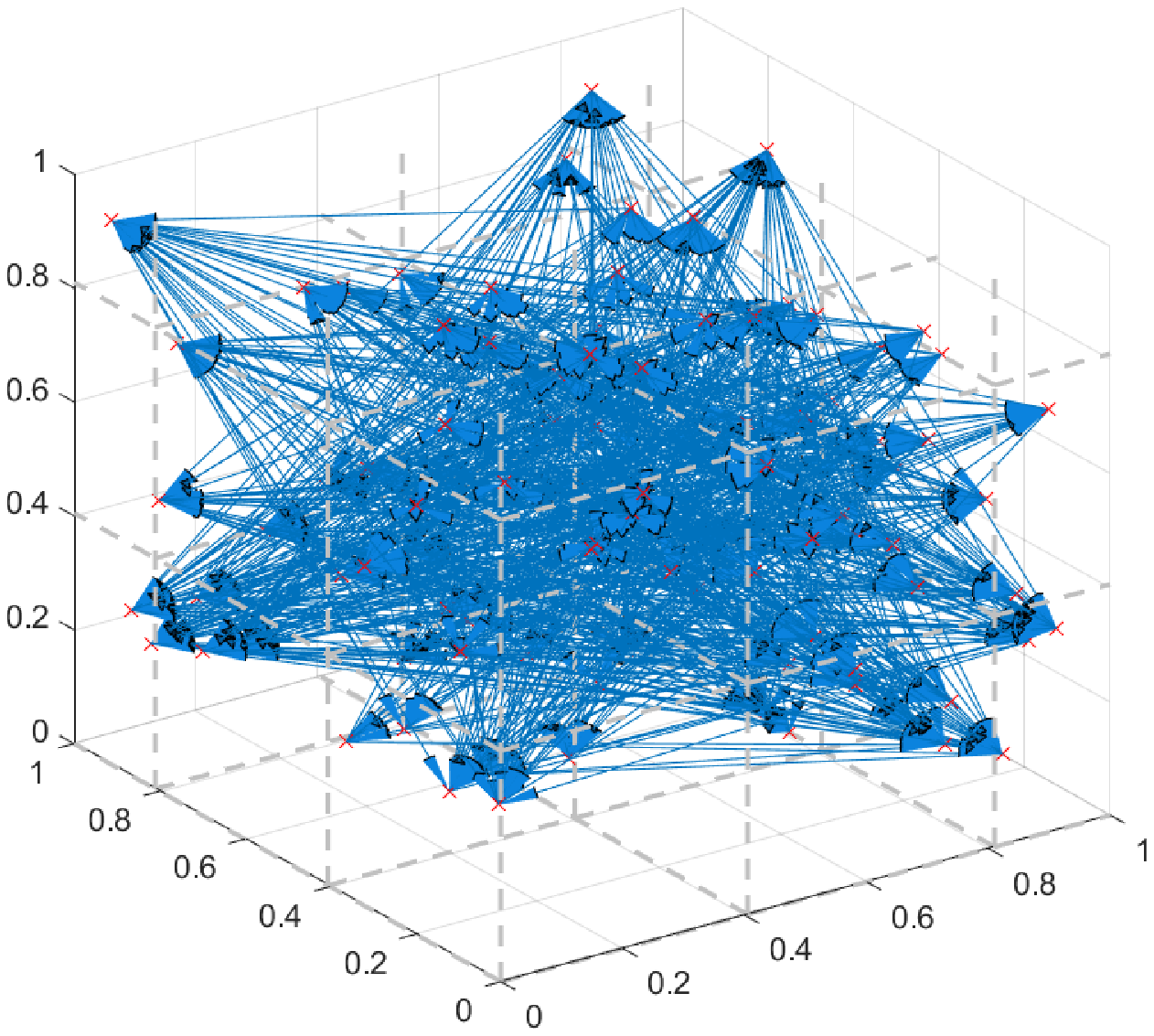}\label{6b}}
	\quad
	\subfigure[Final Communication Relationship]{\includegraphics[width=0.45\textwidth]{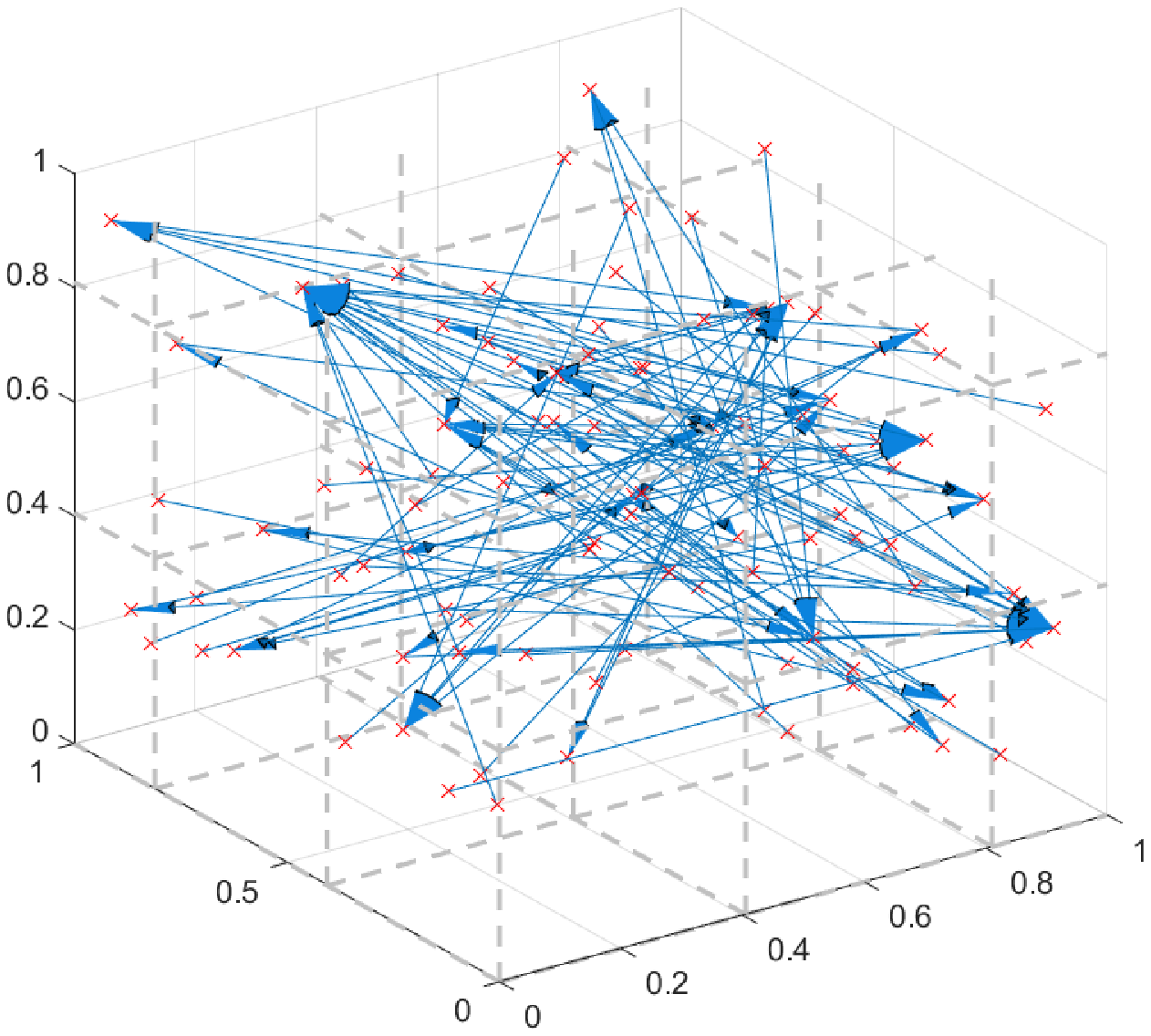}\label{6c}}
	\caption{The contact group selection, the communication selection in the contact group, and the final communication relationship}
	\label{Add_pic_2}
\end{figure*}

{ Then, according to the $L$-routing scheme, whether a transmission adopts ad hoc mode or cellular mode is determined. Taking the number of nodes $n$ as the variable, the results of average hops and throughput under different parameters of $\alpha$, $\beta$ and $\gamma$ are simulated, which is shown in Fig. \ref{AverageHops} and Fig. \ref{ThroughputCapacity}.}

Fig. \ref{AverageHops} illustrates the average number of hops of ad hoc flows under different parameters $\alpha$ and $\beta$.
Fig. \ref{ThroughputCapacity} shows the throughput of the ad hoc flows.
$L$ is selected in the optimal range to maximize the throughput of ad hoc flows, and the bandwidth of ad hoc mode is set to be $W_a = 1$.
There are the following observations.

\emph{1)} When the values of $\alpha$ and $\beta$, or the sum of them increase, the average number of hops of flows decreases correspondingly, and the throughput of UAV network increases.
This is due to the fact that $\alpha$ and $\beta$ affect the location distribution of the destination nodes from the source node.
When $\alpha$ or $\beta$ is large, or the sum of them exceeds a certain range, the destination nodes will be highly clustered around the source node, so that the average number of hops of ad hoc flows are reduced, and the network throughput increases accordingly.

\emph{2)} Within a certain range of $\alpha$ and $\beta$, the size of $L$ will affect the type of the dominant nodes. 
For example, when $0 \le \alpha<3$, $0 \le \beta <3$ and $\alpha + \beta >4$, the simulation results of the average number of hops and throughput are in good agreement with the theoretical results, which shows that the flows of normal nodes is dominant in the network.

\begin{figure}[!ht]
	\centering
	\subfigure[The relationship between the number of nodes and the average number of hops of ad hoc network]{\includegraphics[width=0.49\textwidth]{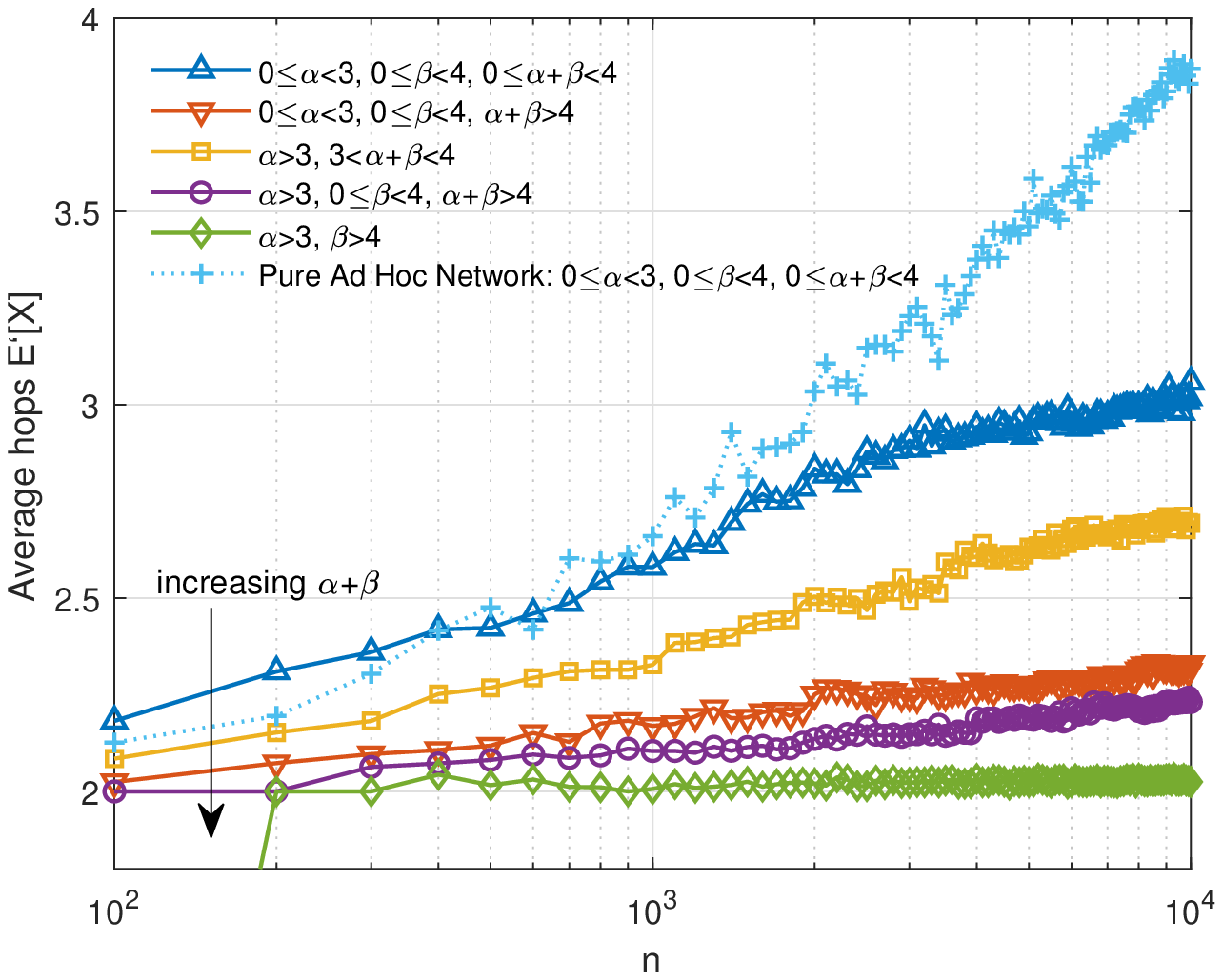}\label{AverageHops}}
	\quad
	\subfigure[The relationship between the number of nodes and throughput of ad hoc network]{\includegraphics[width=0.49\textwidth]{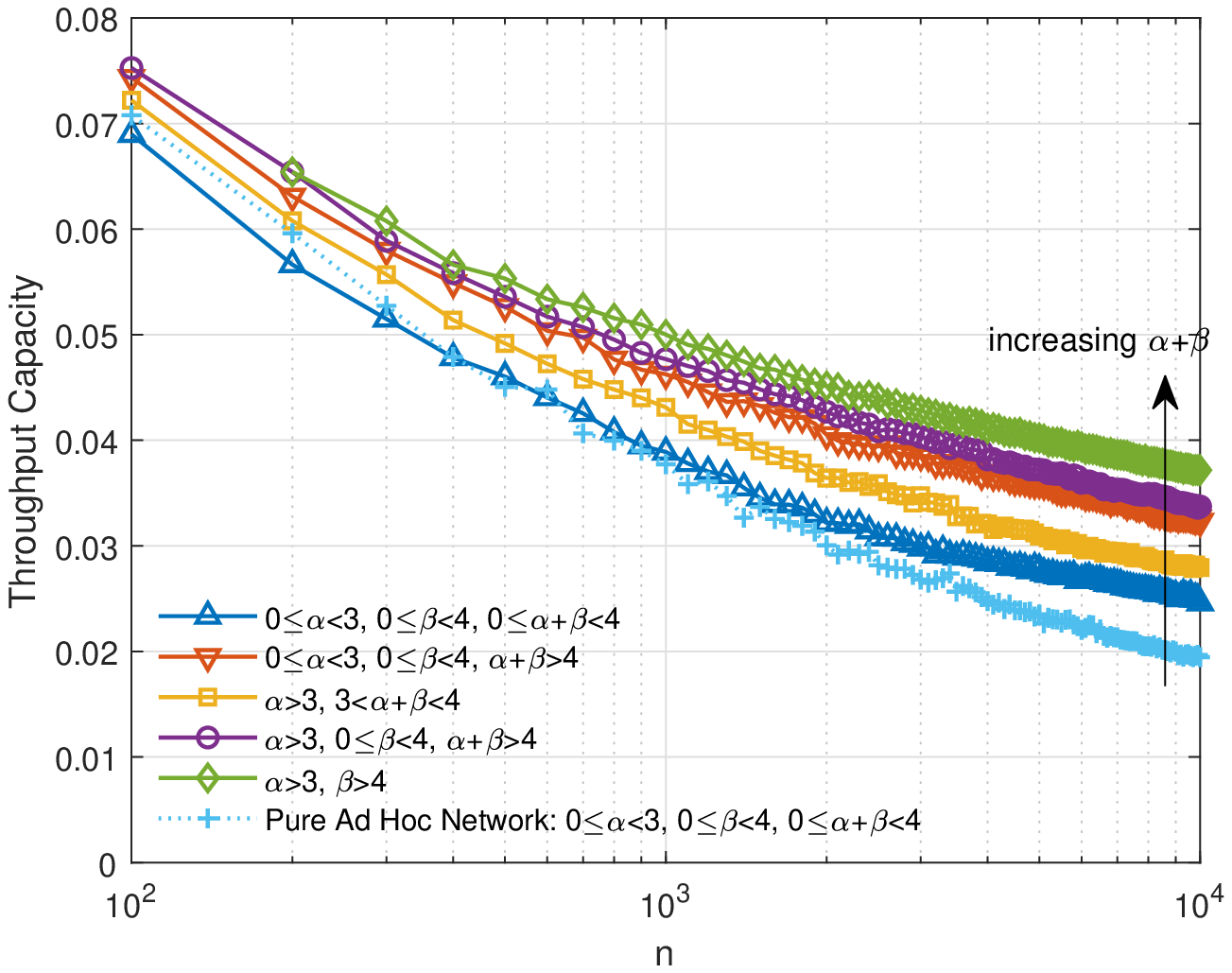}\label{ThroughputCapacity}}
	\caption{The relationship between parameters of UAV network}
	\label{Fig_Final}
\end{figure}

\emph{3)} The introduction of the cellular transmission mode improves the throughput of the scale-free UAV network, compared with the throughput of pure ad hoc network studied in \cite{Intro19}.
Fig. \ref{AverageHops} shows that as $n$ increases, the average number of hops of ad hoc flows in hybrid UAV
 network is smaller than that of pure ad hoc network.
Correspondingly, Fig. \ref{ThroughputCapacity} shows that as $n$ increases, the throughput of hybrid UAV network is higher than that of pure ad hoc network.
This is due to the fact that for the flows with the number of hops larger than $L$, the nodes will directly connect to BSs and exploit the resources of cellular network for transmission. 
Therefore, the number of ad hoc flows is reduced, and the resources of ad hoc network are saved, so that the network throughput is improved.

\section{Conclusion}
In this paper, aiming at improving the throughput of UAV network, the hybrid UAV network with scale-free topology is studied.
Besides, the impact of various parameters on the network throughput is analyzed.
The optimal hop threshold $L$ for the selection of ad hoc or cellular transmission mode is derived, which is a function of the number of nodes and scale-free parameters.
This paper will provide guidance for the architecture design and protocol design for the future UAV network.

\section*{Appendix A}
	According to \cite{Intro19} and law of large numbers (LLN), we have
\begin{equation}\small
	\frac{d_{k}^{-\alpha} \sigma_{q-1}\left(\mathbf{d}_{\mathbf{n}}^{\overline{\mathbf{k}}}\right)}{\sigma_{q}\left(\mathbf{d}_{\mathbf{n}}\right)}=\frac{q}{n}.
\end{equation}

Therefore, (\ref{eq_21}) and (\ref{eq_22}) can be simplified as
\begin{equation}\small
	\label{eq_26}
	\operatorname{Pr}_{1}^{a}=\sum_{x=1}^{L} \sum_{l=1}^{4 x^{2}+2} \sum_{o_{k} \in c_{l}} \sum_{q=q_{0}+1}^{n-1} \frac{q^{-\gamma+1} d_{k}^{-\beta}}{n \sigma_{1}(\mathbf{q}) \sigma_{1}\left(\mathbf{d}_{\mathbf{q}}\right)}.
\end{equation}
\begin{equation}\small
	\label{eq_27}
	\operatorname{Pr}_{1}^{c}=\sum_{x=L+1}^{r^{-1}(n)} \sum_{l=1}^{4 x^{2}+2} \sum_{o_{k} \in c_{l}} \sum_{q=q_{0}+1}^{n-1} \frac{q^{-\gamma+1} d_{k}^{-\beta}}{n \sigma_{1}(\mathbf{q}) \sigma_{1}\left(\mathbf{d}_{\mathbf{q}}\right)}.
\end{equation}

According to LLN, we have
\begin{equation}\small
	\frac{1}{q} \sigma_{1}\left(\mathbf{d}_{\mathbf{q}}\right)=E\left[\mathbf{d}_{\mathbf{q}}\right].
\end{equation}

Therefore,
\begin{equation}\label{eq_sigmadq}\small
	\sum_{q=q_{0}}^{n-1} \frac{q^{-\gamma+1}}{\sigma_{1}\left(\mathbf{d}_{\mathbf{q}}\right)}=\frac{1}{E\left[\mathbf{d}_{\mathbf{q}}\right]} \sum_{q=q_{0}}^{n-1} q^{-\gamma}.
\end{equation}

According to (\ref{E2}), we have
\begin{equation}\label{eq_28}\small
	\begin{aligned}
		E\left[\mathbf{d}_{\mathbf{q}}\right] \equiv E\left[\left\{d_{g_{1}}^{-\beta}, d_{g_{2}}^{-\beta}, \ldots, d_{g_{q}}^{-\beta}\right\}\right],
	\end{aligned}
\end{equation}
where $d_{g_{i}} (i=1,2,\dots, q)$ is the the distance between source and destination, which can be replaced by $xr(n)$, so we have

\begin{equation}\label{eq_29}\small
	\begin{aligned}
		E\left[\mathbf{d}_{\mathbf{q}}\right] &\equiv  \sum_{x=1}^{r^{-1}(n)} \operatorname{Pr}(X=x)(x r(n))^{-\beta}\\
		& = r(n)^{-\beta}.
	\end{aligned}
\end{equation}

Therefore, with (\ref{eq_sigmadq}) and (\ref{eq_29}), ${\rm Pr}_{1}^{a}$ can be simplified as
\begin{equation}\small
	\operatorname{Pr}_{1}^{a} \equiv \frac{r^{\beta}(n)}{n} \sum_{x=1}^{L} \sum_{l=1}^{4 x^{2}+2} \sum_{o_{k} \in c_{l}} d_{k}^{-\beta} \sum_{q=q_{0}+1}^{n-1} \frac{q^{-\gamma}}{\sigma_{1}(\mathbf{q})}.
\end{equation}
If $\gamma > 1$, $\sum_{q=q_{0}}^{n-1} q^{-\gamma}$ and $\sigma_{1}(\mathbf{q})$ are all partial sum of the Riemann Zeta function, which has the following relations.
\begin{equation}\small
	\label{eq_30}
	\sum_{q=q_{0}}^{n-1} q^{-\gamma} \leq \sigma_{1}(\mathbf{q}) \leq \zeta(\gamma) \equiv \Theta(1).
\end{equation}

Therefore, when $\gamma > 1$, according to (\ref{eq_30}), $q^{-\gamma}/\sigma_{1}(\mathbf{q}) = 1$. $\operatorname{Pr}_{1}^{a}$ can be derived as follows
\begin{equation}\small
	\label{eq_32}
	\operatorname{Pr}_{1}^{a} \equiv \frac{r^{\beta}(n)}{n} \sum_{x=1}^{L} \sum_{l=1}^{4 x^{2}+2} \sum_{o_{k} \in c_{l}} d_{k}^{-\beta}.
\end{equation}

For $0 \le \gamma \le 1$, it's obvious that
\begin{equation}\small
	\sum_{q=q_{0}+1}^{n-1} q^{-\gamma} = \Theta (\sigma_{1}(\mathbf{q})) = \Theta(n^{1-\gamma} /(1-\gamma)).
\end{equation}

Thus, $\operatorname{Pr}_{1}^{a}$ is still equivalent to (\ref{eq_32}) when $0 \le \gamma \le 1$, i.e., $\operatorname{Pr}_{1}^{a}$ has the same form when $\gamma$ varies.

In (\ref{eq_32}), $\sum_{l=1}^{4 x^{2}+2}\sum_{o_{k} \in c_{l}} (\cdot)$ represents the number of nodes in the small cubes with $x$ hops {on average}, which has the same order as
\begin{equation}\small
	\label{eq_33}
	N n(r(n))^{3}=\left(4 x^{2}+2\right) n(r(n))^{3}.
\end{equation}

Therefore, with (\ref{eq_33}) and Riemann integral, we have
\begin{equation}\small
	\begin{aligned}
		\operatorname{Pr}_{1}^{a} & \equiv \frac{r^{\beta}(n)}{n} \sum_{x=1}^{L} \sum_{l=1}^{4 x^{2}+2} \sum_{o_{k} \in c_{l}} d_{k}^{-\beta} \equiv(r(n))^{3} \sum_{x=1}^{L}\left(x^{2-\beta}+x^{-\beta}\right) \\
		& \equiv(r(n))^{3} \int_{1}^{L}\left(v^{2-\beta}+v^{-\beta}\right) d v.
	\end{aligned}
\end{equation}

Thus, the simplified form of $\mathrm{Pr}_{1}^{a}$ is as follows.
\begin{equation}\small
	\operatorname{Pr}_{1}^{a} \equiv\left\{\begin{array}{ll}
		\Theta\left(r^{3}(n) L^{3-\beta}\right) & 0 \leq \beta<3 \\
		\Theta\left(r^{3}(n) \ln L\right) & \beta=3 \\
		\Theta\left(r^{3}(n)\right) & \beta>3
	\end{array}\right.
\end{equation}

The simplified form of $\mathrm{Pr}_{1}^{c}$ can be derived similarly, which is
\begin{equation}\small
	\operatorname{Pr}_{1}^{c} \equiv\left\{\begin{array}{ll}
		\Theta\left(r^{\beta}(n)-r^{3}(n) L^{3-\beta}\right) & 0 \leq \beta<3 \\
		\Theta\left(r^{3}(n) \ln \left(\frac{1}{\operatorname{Lr}(n)}\right)\right) & \beta=3 \\
		\Theta\left(r^{3}(n)\right) & \beta>3
	\end{array}\right.
\end{equation}

\section*{Appendix B}
In (\ref{eq_pr2a}) and (\ref{eq_pr2c}), the term $d_k^{- \alpha }{\sigma _{q - 1}}\left( {{\bf{d}}_{\bf{n}}^{\overline {\bf{k}} }} \right)$ can be expanded as follows.

\begin{equation}\small
	\begin{aligned}
		d_k^{- \alpha }{\sigma _{q - 1}}\left( {{\bf{d}}_{\bf{n}}^{\overline {\bf{k}} }} \right) & = d_k^{- \alpha }\left( {{\sigma _{q - 1}}\left( {{{\bf{d}}_{\bf{n}}}} \right) - d_k^{- \alpha }{\sigma _{q - 2}}\left( {{\bf{d}}_{\bf{n}}^{\overline {\bf{k}} }} \right)} \right)\\
		& \le d_k^{- \alpha }{\sigma _{q - 1}}\left( {{{\bf{d}}_{\bf{n}}}} \right).
	\end{aligned}
\end{equation}

Hence, the upper bound of $d_k^{- \alpha }{\sigma _{q - 1}}\left( {{\bf{d}}_{\bf{n}}^{\overline {\bf{k}} }} \right)$ is $d_k^{- \alpha }{\sigma _{q - 1}}\left( {{{\bf{d}}_{\bf{n}}}} \right)$.
According to Lemma 4 in \cite{16}, when $q \le {q_0}$, we have

\begin{equation}\small
	\frac{{{\sigma _{q - 1}}\left( {{{\bf{d}}_{\bf{n}}}} \right)}}{{{\sigma _q}\left( {{{\bf{d}}_{\bf{n}}}} \right)}} \equiv \frac{1}{{{\sigma _1}\left( {{{\bf{d}}_{\bf{n}}}} \right)}}\Theta \left( {\frac{{nq}}{{n - q + 1}}} \right),
\end{equation}

and

\begin{equation}\small
	\Theta \left( {\frac{{nq}}{{n - q + 1}}} \right) = \Theta \left( q \right) = \Theta \left( 1 \right).
\end{equation}

Hence, when $q \le {q_0}$, $\Pr _2^a$ is equivalent to

\begin{equation}\label{eq_49}\small
	\begin{aligned}
		& \operatorname{Pr} _2^a \equiv \sum\limits_{x = 1}^L {\sum\limits_{l = 1}^{4{x^2} + 2} {\sum\limits_{{o_k} \in {c_l}} {\sum\limits_{q = 1}^{{q_0}} {\frac{{{q^{- \gamma }}d_k^{- \alpha  - \beta }{\sigma _{q - 1}}\left( {{{\bf{d}}_{\bf{n}}}} \right)}}{{{\sigma _1}\left( {\bf{q}} \right){\sigma _1}\left( {{{\bf{d}}_{\bf{q}}}} \right){\sigma _q}\left( {{{\bf{d}}_{\bf{n}}}} \right)}}} } } } \\
		& \equiv \sum\limits_{x = 1}^L {\sum\limits_{l = 1}^{4{x^2} + 2} {\sum\limits_{{o_k} \in {c_l}} {\frac{{d_k^{- \alpha  - \beta }}}{{{\sigma _1}\left( {\bf{q}} \right){\sigma _1}\left( {{{\bf{d}}_{\bf{n}}}} \right)}}\sum\limits_{q = 1}^{{q_0}} {\frac{{{q^{- \gamma }}}}{{{\sigma _1}\left( {{{\bf{d}}_{\bf{q}}}} \right)}}} } } }.
	\end{aligned}
\end{equation}

According to (44) in \cite{Intro17}, we have
\begin{equation}\small
	\sum\nolimits_{q = 1}^{{q_0}} {{{{q^{- \gamma }}} \mathord{\left/
				{\vphantom {{{q^{- \gamma }}} {{\sigma _1}\left( {{{\bf{d}}_{\bf{q}}}} \right)}}} \right.
				\kern-\nulldelimiterspace} {{\sigma _1}\left( {{{\bf{d}}_{\bf{q}}}} \right)}} \equiv {r^\beta }\left( n \right)},
\end{equation}
which is substituted into (\ref{eq_49}).
Using integral transformation techniques, we have

\begin{equation}\label{eq_420}\small
	\begin{aligned}
		& \operatorname{Pr} _2^a \equiv \frac{{{r^\beta }\left( n \right)}}{{{\sigma _1}\left( {\bf{q}} \right){\sigma _1}\left( {{{\bf{d}}_{\bf{n}}}} \right)}}\sum\limits_{x = 1}^L {\sum\limits_{l = 1}^{4{x^2} + 2} {\sum\limits_{{o_k} \in {c_l}} {d_k^{- \alpha  - \beta }} } } \\
		& {\rm{    }} \equiv \frac{{n{r^{3 - \alpha }}\left( n \right)}}{{{\sigma _1}\left( {\bf{q}} \right){\sigma _1}\left( {{{\bf{d}}_{\bf{n}}}} \right)}}\sum\limits_{x = 1}^L {\left( {{x^{2 - \alpha  - \beta }} + {x^{- \alpha  - \beta }}} \right)} \\
		& {\rm{    }} \equiv \frac{{n{r^{3 - \alpha }}\left( n \right)}}{{{\sigma _1}\left( {\bf{q}} \right){\sigma _1}\left( {{{\bf{d}}_{\bf{n}}}} \right)}}\int_1^L {\left( {{\upsilon ^{2 - \alpha  - \beta }} + {\upsilon ^{- \alpha  - \beta }}} \right)} d\upsilon.
	\end{aligned}
\end{equation}

According to (16) in \cite{Intro19}, we have

\begin{equation}\label{eq_44}\small
	{\sigma _1}\left( {{{\bf{d}}_{\bf{n}}}} \right) \equiv \left\{{\begin{array}{*{20}{l}}
			{\Theta \left( n \right){\rm{                   }}0 \le \alpha  < 3}\\
			{\Theta \left( {n\ln \left( {{r^{- 1}}\left( n \right)} \right)} \right){\rm{ }}\alpha  = 3}\\
			{\Theta \left( {n{r^{3 - \alpha }}\left( n \right)} \right){\rm{       }}\alpha  > 3}
	\end{array}} \right.
\end{equation}

Besides, there is the following relation.

\begin{equation}\label{eq_45}\small
	\begin{aligned}
		& {\sigma _1}\left( {\bf{q}} \right) = \sum\nolimits_{q = 1}^{n - 1} {{q^{- \gamma }}}
		& \equiv \left\{{\begin{array}{*{20}{l}}
				{\Theta \left( 1 \right){\rm{      }}\gamma  > 1}\\
				{\Theta \left( {{n^{1 - \gamma }}} \right){\rm{ }}0 \le \gamma  \le 1}
		\end{array}} \right.
	\end{aligned}
\end{equation}

Substituting (\ref{eq_420}) and (\ref{eq_44}) into (\ref{eq_45}), the values of $\Pr _2^a$
are revealed in {Table \ref{tab_Pr2a}}. Similarly, using the techniques of 
integral transformation, the values of $\Pr _2^c$
are revealed in {Table \ref{tab_Pr2c}}.

\end{document}